\documentclass[letterpaper,twocolumn,10pt]{article}
\usepackage{usenix-2020-09}

\usepackage{amsmath}

\newcommand{\system}{Lotto}

\usepackage{xurl}

\usepackage{cleveref}
\crefformat{section}{§#2#1#3}

\usepackage{subcaption}

\newcommand{\PHB}[1]{\noindent\textbf{#1}\hspace{.5em}} 
\newcommand{\PHM}[1]{\vspace{.2em}\noindent\textbf{#1}\hspace{.5em}} 

\usepackage[T1]{fontenc}
\usepackage{inconsolata}

\usepackage{bm}

\usepackage{amsthm}
\usepackage{amsfonts}
\usepackage[thinc]{esdiff}

\newtheorem{theorem}{Theorem}
\newtheorem{corollary}{Corollary}

\makeatletter
\newtheorem*{rep@theorem}{\rep@title}
\newcommand{\newreptheorem}[2]{%
	\newenvironment{rep#1}[1]{%
		\def\rep@title{#2 \ref{##1}}%
		\begin{rep@theorem}}%
		{\end{rep@theorem}}}
\makeatother
\newreptheorem{theorem}{Theorem}
\newreptheorem{lemma}{Lemma}
\newreptheorem{corollary}{Corollary}

\definecolor{myred}{rgb}{0.58, 0.06, 0}

\usepackage{contour}
\usepackage[normalem]{ulem}

\contourlength{0.8pt}

\newcommand{\client}[1]{\textcolor{myred}{\underline{#1}}}

\usepackage{soul}
\definecolor{mygreen}{rgb}{0.88, 0.91, 0.89}
\sethlcolor{mygreen}
\newcommand{\server}[1]{\hl{#1}}

\newcommand{\informed}[1]{\textit{[#1]}}

\usepackage{enumitem}
\setlist{leftmargin=5mm, itemsep=0mm}

\usepackage{amsmath,graphicx}
\newcommand{\widesim}[3][3.5]{
  \mathrel{\overunderset{#2}{#3}{\scalebox{#1}[1]{$\sim$}}}
}

\usepackage{dsfont}  

\usepackage{balance}

\usepackage{booktabs}
\usepackage{tabularx}
\usepackage{multirow}
\usepackage{amssymb}

\usepackage{ulem}
\definecolor{cadmiumgreen}{rgb}{0.0, 0.42, 0.24}
\definecolor{myblue}{RGB}{0, 11, 172}

\usepackage[absolute,overlay]{textpos}

\usepackage{soul}

\begin{document}

\date{}

\title{\Large \system{}: Secure Participant Selection against Adversarial Servers\\in Federated Learning}

\author{
	\rm{
		Zhifeng Jiang$^{\text{1}}$ \enskip
		Peng Ye$^{\text{1}}$ \enskip
		Shiqi He$^{\text{2}, *}$ \enskip
		Wei Wang$^{\text{1}}$ \enskip
		Ruichuan Chen$^{\text{3}}$ \enskip
		Bo Li$^{\text{1}}$
	}\\
	{
		$^{\text{1}}$HKUST \enskip
		$^{\text{2}}$University of Michigan \enskip
		$^{\text{3}}$Nokia Bell Labs
	}
}

\pagestyle{empty}  
\maketitle


\begin{abstract}
    
In Federated Learning (FL), many privacy-enhancing techniques, such as secure
aggregation and distributed differential privacy, provide security guarantees
under the assumption of having an \emph{honest majority} among participants.
However, an adversarial server can strategically select compromised clients
to create a dishonest majority, thereby undermining the system's security
guarantees. In this paper, we present \system{}, an FL system that addresses
this fundamental, yet underexplored issue by providing secure participant
selection in the presence of an adversarial server. \system{} supports two selection
algorithms: \emph{random} and \emph{informed}. To ensure random selection
without a trusted server, \system{} enables each client to autonomously
determine their participation using \emph{verifiable randomness}. For
informed selection, which is more vulnerable to manipulation, \system
{} approximates the algorithm by employing random selection within a \emph
{refined client pool}. Theoretical analysis shows that \system
{} effectively aligns the proportion of server-selected compromised
participants with the base rate of dishonest clients in the population.
Large-scale experiments further reveal that \system{} achieves
time-to-accuracy performance comparable to that of insecure selection
methods.

\end{abstract}
{\let\thefootnote\relax\footnote{{$^*$Work done while the author was doing research internship at HKUST.}}}
\section{Introduction}
\label{sec:intro}

Edge devices, such as smartphones and laptops, are becoming increasingly
powerful and can collaborate to build high-quality machine learning
(ML) models using data collected on the devices~\cite{kairouz2019advances}. To
protect data privacy, major companies like Google and Apple have
adopted \emph{federated learning} (FL)~\cite{mcmahan2017communication} for
tasks in computer vision (CV) and natural language processing (NLP) across
client devices~\cite{granqvist2020improving, song2022flair, xu2023federated}.
In FL, a server dynamically samples a small subset of clients from a large
population in each training round. The sampled clients, aka \emph
{participants}, compute individual local updates using their own data and
upload only the updates to the server for global aggregation, without
revealing the training data~\cite{bonawitz2019towards}.

However, simply keeping client data undisclosed is insufficient to
preserve data privacy. Recent studies have shown that sensitive data
can be leaked through individual updates during FL training~\cite
{geiping2020inverting, yin2021see, wang2022protect, zhao2023secure, yue2023gradient}, and it
is also possible to infer a client's data from trained models~\cite
{shokri2017membership, carlini2019secret, song2019auditing,
nasr2019comprehensive} by exploiting their ability to memorize
information~\cite{carlini2019secret, song2019auditing}. To minimize data
leakage, current FL systems use \emph{secure aggregation} techniques~\cite
{bonawitz2017practical, bonawitz2019towards, bell2020secure} to ensure that
adversaries only learn the aggregated updates, not individual ones.
Additionally, these systems may employ \emph{distributed differential privacy}
(DP)~\cite{kairouz2021distributed, agarwal2021skellam, stevens2022efficient}
to limit the extent of data leakage from the aggregated update (\cref
{sec:background_defense}). Since the server may not always be trusted, these
privacy-enhancing techniques commonly assume that the \emph{majority} of the
participants act honestly and can jointly detect and prevent server misbehavior.

Nonetheless, this assumption may not hold in practice -- an adversarial
server can strategically select compromised clients with whom it colludes, resulting
in a practical attack. While the overall client population is
large, the number of participants selected in each round is often limited to a
few hundred for efficiency reasons~\cite{bonawitz2019towards}. Consequently,
even if the majority of the population is honest (e.g., 99\% out of 100,000
clients), the server can still select enough adversaries to
create a dishonest majority among participants.

The presence of a dishonest majority 
undermines the downstream privacy-enhancing techniques, leading to
significant security vulnerabilities. As shown in \cref{sec:background_loophole}, the
privacy guarantees provided by distributed DP degrade rapidly as the
dishonest cohort expands~\cite{kairouz2021distributed, agarwal2021skellam,
stevens2022efficient}. Additionally, secure aggregation can only tolerate a
certain proportion of clients colluding with the server; exceeding this
limit would enable the adversary to reconstruct a client's update~\cite
{bonawitz2017practical}. It is hence crucial to prevent the server from maliciously manipulating the client selection process.

Despite its importance, this problem remains largely unaddressed in the literature. The main
challenge arises from the use of a \emph{central server} in the cross-device
FL scenarios~\cite{kairouz2019advances}. In the absence of efficient
peer-to-peer communication across a vast array of heterogeneous devices~\cite
{yang2021characterizing, lai2022fedscale}, the server acts as a central hub for
facilitating message exchanges between clients. Since the server is untrusted,
an honest client lack a reliable view of other clients, such as their
availability, utility, and integrity. Consequently, it is difficult for a client to
verify whether the server has correctly executed the prescribed participant
selection algorithm, as this requires the true knowledge of and/or
inputs from other clients.

In this paper, we present \system{}, a novel framework that enables secure
participant selection in FL for the first time. \system{} aims to provide the
following desired property: \textbf{among the selected participants, the
proportion of the compromised ones approximately aligns with the base rate of
dishonest clients in the overall population}. \system{} uses \emph
{verifiable randomness} in the selection process.

Specifically, to achieve secure random selection~\cite
{mcmahan2017communication}, \system{} allows each client in the population to
determine its participation in a training round by computing \emph
{verifiable random functions} (VRFs)~\cite
{dodis2005verifiable, micali1999verifiable} using its secret key and public
inputs. When a client chooses to participate, \system{} collects the generated
randomness and the associated VRF proof, which can be verified by other participants to
ensure the integrity and validity of the randomness. This approach ensures that each
client's participation is provably independent, fair, and unpredictable
for other clients. However, it is still possible for an adversarial server to 
exploit this process and increase the number of dishonest clients by, for example, 
sending incorrect and/or a selectively chosen messages.
To mitigate such misconduct, \system{} lets honest clients verify critical messages,
adding an additional layer of security.
\system{} \emph{provably ensures} that the compromised fraction of selected participants remains close to the base rate in the population (\cref{sec:design_random}).

\system{} also provides security guarantee for more advanced
algorithms, referred to as \emph{informed selection}, which are commonly used in FL
to select the best-performing participants for optimal training
efficiency~\cite{zhang2021client, nishio2019client, wang2020optimizing,
chai2020tifl, lai2021oort, kim2021autofl}. These algorithms rely on
client-specific measurements, such as processing speed and data quality, which
are difficult to verify. Instead of precisely following the
selection logic of these algorithms, \system{} \emph{approximates} them by transforming
an informed selection problem into a random one. To achieve this, it lets the server
\emph{refine} the client pool by excluding a small fraction of low-quality
clients based on the specified measurements. \system{} then performs secure
random selection within the refined pool. Although an adversarial server may
attempt to manipulate this process by excluding honest clients, the security
of \system{}'s random selection ensures that any advantage it gains is provably
small (\cref{sec:design_informed}).

We have implemented \system{} as a library\footnote{\system{} is available at \url{https://github.com/SamuelGong/Lotto}.} (\cref{sec:impl}) and evaluated its
performance with various FL tasks in a large EC2 cluster configured to
emulate a cross-device scenario (\cref{sec:eval}). Compared to insecure
selection algorithms, \system{} slightly increases the duration of each
training round by less than 10\% while achieving comparable or even better
time-to-accuracy performance and inducing negligible network cost.

\section{Background and Motivation}
\label{sec:background}

In this section, we begin by providing an overview of the standard FL
workflow (\cref{sec:background_fl}). We then discuss the privacy concerns and
prevalent privacy-preserving approaches that critically rely on an honest
majority among participants (\cref{sec:background_defense}). Following this,
we highlight the privacy vulnerabilities that arise when encountering a
dishonest majority (\cref{sec:background_loophole}) and ultimately present
the motivation for our work.

\subsection{Federated Learning}
\label{sec:background_fl}

Federated learning (FL)~\cite
{mcmahan2017communication,kairouz2019advances} emerges as a new private
computing paradigm that enables a large number of \emph{clients} to
collaboratively train a global model under the orchestration of a \emph
{server}. The standard FL workflow is an iterative process that comprises
three steps in each training round~\cite{bonawitz2019towards}. \textcircled
{\raisebox{-1.0pt}{1}} Participant Selection: the server samples a subset of
eligible clients from the entire \emph{population} as participants, where the
eligibility criteria may include factors such as charging status and whether being connected to an unmetered network. \textcircled{\raisebox{-1.0pt}
{2}} Local Training: the server distributes the global model to each
participant, who then trains the model using its private data to compute
a \emph{local update}. \textcircled{\raisebox{-1.0pt}{3}} Model Aggregation:
the server collects local updates from participants and computes an \emph
{aggregated update}. This step is often coupled with additional privacy-enhancing techniques (\cref{sec:background_defense}). The server then uses
the aggregated update to refine the global model. As participants only expose
their local updates instead of raw data to other parties, FL adheres to the
data minimization principle~\cite
{kairouz2019advances, cummings2023challenges} mandated by many privacy
regulations~\cite{voigt2017eu}. Consequently, FL has been deployed to
facilitate a broad range of edge intelligence applications~\cite
{granqvist2020improving, song2022flair, xu2023federated}.

\subsection{Enhancing Data Privacy in FL}
\label{sec:background_defense}

Despite the fact that client data is not directly revealed during the FL training process, extensive research has demonstrated that it is still possible to disclose sensitive data using individual updates or trained models.
Specifically, an adversary can reconstruct a client's training data from its local updates through \emph{data extraction}~\cite{geiping2020inverting, yin2021see, wang2022protect, zhao2023secure, yue2023gradient}.
For image classification tasks, even with a huge number of local gradient steps (e.g., 15k) and large models like ResNet-18~\cite{he2016deep}, state-of-the-art attackers~\cite{yue2023gradient} can successfully recover high-quality training images at both the pixel-level and semantic-level using a single local update.
Given only an aggregated update or a trained model, the adversary can also infer whether a client's private text has been used in the training set through \emph{membership inference}~\cite
{shokri2017membership, carlini2019secret, song2019auditing,
nasr2019comprehensive}.

To effectively control the exposure of clients' data
against these attacks, FL systems commonly employ two privacy-enhancing techniques: secure aggregation and distributed differential privacy.

\PHM{Secure Aggregation (SecAgg)~\cite{bonawitz2017practical}.}
SecAgg is a secure multi-party computation (SMPC) protocol that enables the
server to learn only the sum of participants' updates and nothing beyond that.
In SecAgg, each participant $i$ first generates two key pairs $(sk^1_i, pk^1_i)$, $(sk^2_i, pk^2_i)$, where the public keys $(pk_i^1, pk_i^2)$ are shared with others.
Each pair of participants $(i, j)$ then engage in a key agreement 
to derive a shared key $s_{i, j} = \texttt{KA.Agree}(sk^1_i, pk^1_j)$.
Each participant $i$ then uses a pseudorandom number generator to derive a pairwise mask, $\bm{p}_{i, j} = \texttt{PRNG}(s_{i, j})$, for each other participant $j$, and adds it to its update following $\bm{y}_i = \bm{x}_i - \sum_{j<i} \bm{p}_{i, j} + \sum_{j>i} \bm{p}_{i, j}$.
As a result, after all participants' masked updates are aggregated, 
the pairwise masks cancel out, i.e., $\sum \bm{y}_i = \sum \bm{x}_i$.

In case of client dropout, SecAgg removes the pairwise masks of the dropped
participants to ensure that the remaining pairwise masks can still cancel out
when aggregated. To achieve this, each participant $i$ secret-shares its key
$sk^1_i$ to each other participant $j$ using a $t$-out-of-$n$ secret sharing
scheme, where $t$ is the minimum number of participants required to recover
the shared key. Therefore, if participant $i$ later drops out, a sufficient number of remaining clients ($\geq t$) can help the server reconstruct its key $sk^1_i$,
thereby enabling the recovery of the missing pairwise masks.

A security concern remains in that the server can falsely claim
the dropout of participant $i$ to recover its pairwise masks: removing these masks,
the server can uncover participant $i$'s local update $x_i$.
To address this concern, SecAgg introduces an additional mask $\bm{q}_i = \texttt{PRNG}(b_i)$ where $b_i$ is independently generated by participant $i$ 
and also secret-shared with each other.
The modified masking process for participant $i$ becomes:
\begin{equation}
    	\bm{y}_i = \bm{x}_i  \underbrace{-\sum_{j<i} \bm{p}_{i, j} + \sum_{j>i} \bm{p}_{i, j}}_{A} + \underbrace{\bm{q}_i}_{B}.
    \label{eq:secagg_complete}
\end{equation}
By requiring that the server only requests secret shares of either $sk_i^1$ (to recover part $A$ if participant $i$ drops out) or $b_i$ (to recover part $B$ otherwise) from other remaining participants, SecAgg ensures that part $A$ and $B$ cannot be removed simultaneously, and the secrecy of $\bm{x}_i$ is preserved.
Subsequent work, such as SecAgg+~\cite{bell2020secure}, has been proposed to reduce the runtime overhead of SecAgg and increase its scalability.

\PHM{Distributed Differential Privacy (DP)~\cite{agarwal2021skellam,kairouz2021distributed}.}
Simply concealing individual local updates using SecAgg is insufficient as 
client data can still leak through the aggregated updates. Differential privacy~\cite
{cynthia2006differential, dwork2014algorithmic} can prevent this by
ensuring that no specific client's participation significantly increases the
likelihood of any observed aggregated update by potential adversaries. This
guarantee is captured by two parameters, $\epsilon$ and $\delta$. Given any
neighboring training datasets $D$ and $D'$ that differ only in the inclusion
of a single participant's data, an aggregation procedure $M$ is $
(\epsilon, \delta)$-differentially private if, for any given set of output
$R$, $\textrm{Pr}[M(D) \in R] \leq e^{\epsilon} \cdot \textrm{Pr}[M(D^
{'})\in R]+\delta$. Therefore, a change in a participant's contribution yields
at most a multiplicative change of $e^\epsilon$ in the probability of any
output $R$, except with probability $\delta$.

To achieve a privacy goal $(\epsilon_G, \delta_G)$ in an $R$-round training, 
each aggregated update needs to be perturbed by a random noise $r$ calibrated based on $\epsilon_G$, $\delta_G$, $R$, and the sensitivity of the aggregated update $\Delta$ (defined as its maximum change caused by the inclusion of a single client's data).
Given the use of SecAgg, the required level of noise can be achieved without relying on the (untrusted) server by having each participant $i$ add an even share of the noise $r_i$ to its local update, i.e.,

\begin{equation}
    \sum_{i} r_i
    \widesim[4]{(\epsilon_G, \delta_G)\textrm{-DP}}{R, \Delta}
     r.
    \label{eq:dp}
\end{equation}
This DP model is referred to as distributed DP in the literature~\cite{kairouz2019advances},
given its distributed fashion of noise addition.

\subsection{Privacy Loopholes in Common Defenses}~\label{sec:background_loophole}

\PHB{Strong Assumption on Honest Majority in Participants.}
The above privacy-enhancing techniques commonly assume that the \emph{majority} of participants are honest, relying on them to detect and/or prevent the server's misbehavior.
However, this assumption often proves inaccurate due to the necessity of participant selection in FL systems.
While the overall population may be extensive, FL usually involves a limited number of sampled clients at each training round, typically in the range of a few hundred~\cite{bonawitz2019towards,kairouz2019advances}.
This limitation arises because including additional clients beyond a certain threshold only yields marginal improvements in terms of convergence acceleration~\cite{mcmahan2017communication, wang2021field}.
Unfortunately, if the server acts maliciously, it can manipulate the selection process by intentionally selecting compromised clients that it colludes with, resulting in a majority of dishonest participants among those involved.
Without an honest majority among participants, these techniques encounter significant security vulnerabilities.

\PHM{Case \#1: SecAgg.}
SecAgg safeguards a participant $i$'s plaintext update $\bm{x}_i$ from server probing by distributing secret shares of $sk_i^1$ and $b_i$ among other participants and requiring them to disclose only the shares of \emph{either} secret to the server (\cref{sec:background_defense}).
However, this can not be achieved if a malicious server colludes with a sufficiently large fraction of participants.

Consider a scenario with $s$ participants, $x$ of which collude with the server.
Let $t$ denote the threshold used by the $t$-out-of-$n$ secret sharing scheme in SecAgg.
As depicted in Figure~\ref{fig:motivation_collusion_secagg}, when $x \geq 2t - s$, the server can identify two separate sets of honest participants, each containing at least $t - x$ members.
In this case, by announcing participant $i$'s contradictory dropout outcome to either set, the server can collect at least $t-x$ secret shares for both $sk_i^1$ and $b_i$.
Combining them with the shares provided by the $x$ colluding clients, the server can reconstruct $sk_i^1$ and $b_i$, simultaneously.
Ultimately, SecAgg fails to protect the local update of participant $i$, as implied by Equation~\eqref{eq:secagg_complete}.

\begin{figure}[t]
	\centering
	\begin{subfigure}[b]{0.49\columnwidth}
		\centering
		\includegraphics[width=\columnwidth]{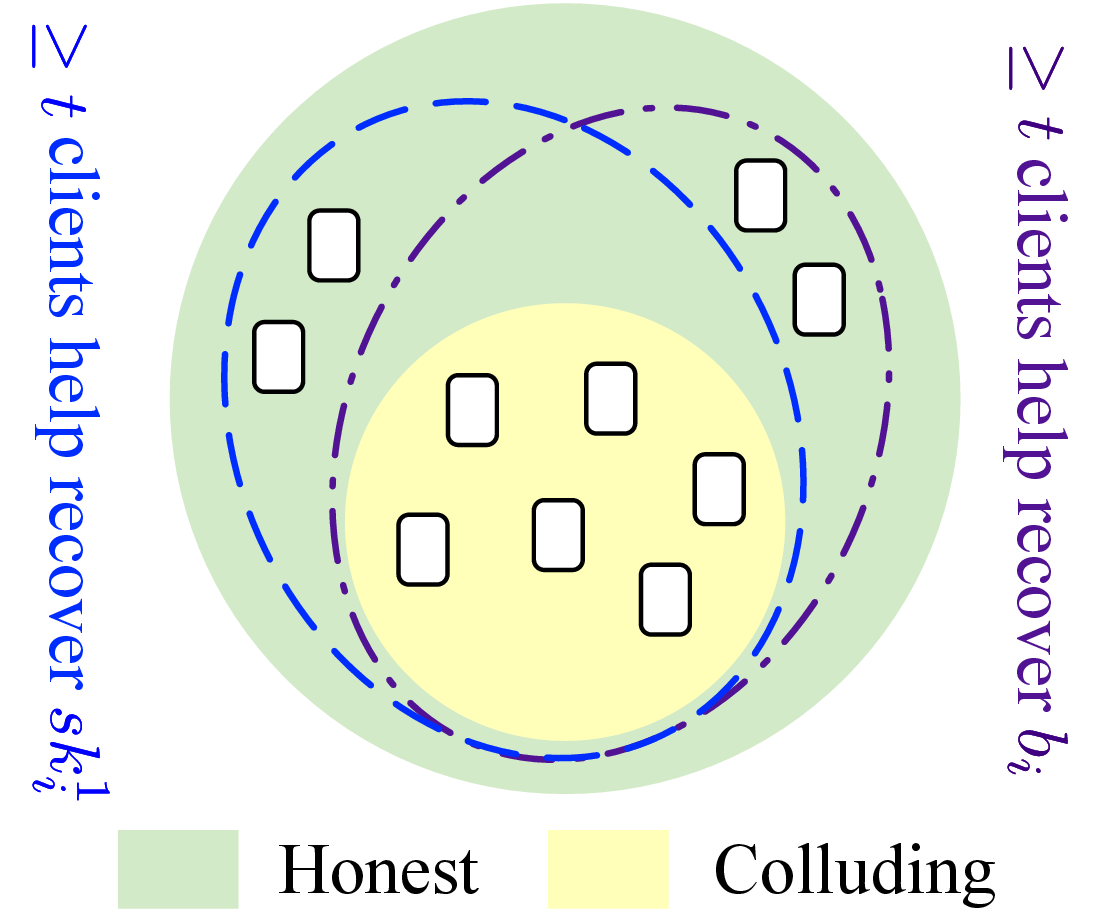}
		\caption{SecAgg.}
		\label{fig:motivation_collusion_secagg}
	\end{subfigure} \hfill
	\begin{subfigure}[b]{0.49\columnwidth}
		\centering
		\includegraphics[width=\columnwidth]{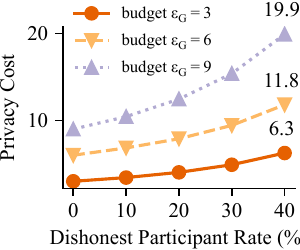}
		\caption{Distributed DP.}
		\label{fig:motivation_collusion_dp}
	\end{subfigure}
	\caption{Example impacts of dishonest majority.}
	\label{fig:motivation_collusion}
\end{figure}

\PHM{Case \#2: Distributed DP.}
Distributed DP, being built upon SecAgg, inherits the aforementioned security vulnerability.
When the number of dishonest participants surpasses a specific threshold ($2t - n$ as mentioned earlier), a malicious server can deduce the precise local update of participant $i$, which is perturbed solely by its noise share $r_i$ (but not $r$).

Even if this threshold is not exceeded, the security of distributed DP can still be impacted by dishonest participants.
If they choose not to add any noise to their local updates, the cumulative noise added to the aggregated update will \emph{fall below} the required level.
This results in increased data exposure, causing the system to exhaust its privacy budget more rapidly than planned.
To demonstrate this issue in a practical FL scenario, we conduct an analysis on a testbed where 700 clients collaborate to train a CNN model over the FEMNIST dataset for 50 rounds.
Each round involves 70 randomly selected clients, with the task of noise addition distributed equally among all participants. However, dishonest participants will abstain from adding any noise.
Figure~\ref{fig:motivation_collusion_dp} illustrates that as the fraction of dishonest participants increases (while remaining below the SecAgg's threshold mentioned earlier), the privacy deficit grows explosively, regardless of the prescribed privacy budget.
For example, when 40\% participants are dishonest, training with a global privacy budget $\epsilon_G=6$ ends up consuming an $\epsilon$ of 11.8 at the 50th round.

\section{\system{}: Secure Participant Selection}~\label{sec:design}

Taking the first stab at the mentioned issues, we present \system{}, a framework that secures participant selection in FL. 
We first formalize the participant selection problem (\cref{sec:design_problem}),  establish the threat model (\cref{sec:design_threat}), and briefly introduce the cryptographic primitives used (\cref{sec:design_primitives}).
Next, we describe how \system{} accomplishes our goal in two scenarios: random selection (\cref{sec:design_random}) and informed selection (\cref{sec:design_informed}).
We finally provide security (\cref{sec:design_security}) analysis for these methods.

\subsection{The Participant Selection Problem}
~\label{sec:design_problem}

In the FL literature, participant selection is commonly approached in two forms, both of which we aim to safeguard.

\PHM{Random Selection.} 
The Federated Averaging (FedAvg) algorithm~\cite{mcmahan2017communication}, considered the canonical form of FL, adapts local SGD~\cite{stich2018local, gorbunov2021local} on a randomly selected subset of clients during each round.
To achieve random selection, the server can perform \emph{subsampling without replacement}, where a subset of the population is chosen uniformly at random.
Alternatively, it can conduct \emph{Bernoulli sampling}, where each client is independently included with the same probability.
We focus on the former as it guarantees a fixed sample size.

\PHM{Informed Selection.}
Given the large heterogeneity in device speed and/or data characteristics, several FL advancements go beyond random sampling by employing informed selection algorithms~\cite{zhang2021client, nishio2019client, wang2020optimizing, chai2020tifl, lai2021oort, kim2021autofl}.
These algorithms leverage \emph{client-specific measurements} to identify and prioritize clients who have the greatest potential to rapidly improve the global model.
Among these measurements, some can be directly monitored by the server, e.g., response latencies, while the others rely on clients' self-reporting, e.g., training losses.

For example, Oort, the state-of-the-art informed algorithm, ranks each client $i$ with a utility score $U_i^{Oort}$ that jointly considers the client's data quality and system speed:

\begin{equation}
	U^{Oort}_i = \underbrace{|B_i| \sqrt{\frac{1}{|B_i|}\sum_{k \in B_i} Loss(k) ^2}}_{\mathrm{Data \  quality}} \times \underbrace{(\frac{T}{t_i})^{\mathds{1}(T<t_i) \times \alpha}}_{\mathrm{System \ speed}}.
	\label{eq:oort_score}
\end{equation}
The first component is the aggregate training loss that reflects the volume and distribution of the client's dataset $B_i$.
The second component compares the client's completion time $t_i$ with the developer-specified duration $T$ and penalizes any delay (which makes the indicator $\mathds{1}(T<t_i)$ outputs 1) with an exponent $\alpha > 0$.
Oort prioritizes the use of clients with high utility scores for enhanced training efficiency.

\PHM{Design Goal.}
Let $n$ represent the population size, with $c$ being the number of dishonest clients, and $s$ the target number of participants in a training round.
Ideally, regardless of the selection approach and the base rate of dishonest clients in the population $c/n$, \ul{the proportion of compromised participants $x/s$ should match $c/n$}, where $x$ represents the number of compromised participants.
For practical purposes, we allow for a slight relaxation of this ideal goal.
Specifically, our objective is to align $x/s$ with $\eta c/n$, where $\eta$ is a small value slightly greater than 1.

\subsection{Threat Model}
~\label{sec:design_threat}

\PHB{Cross-Device Scenarios with Server-Mediated Network.}
Each client has a \emph{private authenticated channel} established with the server, ensuring that their transmitted messages cannot be eavesdropped on by external parties.
We do not assume a peer-to-peer network among clients due to concerns regarding efficiency and scalability.
Instead, clients rely on the server to mediate communication between them.

\PHM{Public Key Infrastructure (PKI).}
We assume a PKI that facilitates the registration of client identities and the generation of cryptographic keys on their behalf.
During the setup phase, each client $i$ \emph{registers} its public key, denoted as $(i, pk_i^{reg}$), to a public bulletin board maintained by the PKI.
The bulletin board only accepts registrations from clients themselves, preventing malicious parties from impersonating honest clients.
Also, the PKI is responsible for \emph{generating} the required public/secret key pair for client $i$ to use in \system{}, $(pk_i^{\system{}}, sk_i^{\system{}})$.
The PKI encrypts the secret key $sk_i^{\system{}}$ using $pk_i^{reg}$ to securely transmit them to client $i$.
Meanwhile, the public key $pk_i^{\system{}}$ is registered on the public bulletin board.
The communication with the PKI is \emph{not} mediated by the server.

\PHM{Malicious Server with Colluding Clients in Population.}
We assume a malicious server colluding with a subset of clients in the population and they can \emph{arbitrarily deviate} from the protocol.
They may send incorrect or chosen messages to honest clients, abort or omit messages, or engage in Sybil attacks~\cite{douceur2002sybil} to increase the portion of clients under their control.
We do \emph{not} assume a specific fraction of dishonest clients within the population to provide fundamental security.
Even in the presence of a dishonest majority in the population, our objective remains to ensure that, among the selected clients, the proportion of the compromised ones aligns with the base rate of dishonest clients in the population (\cref{sec:design_problem}).
In other words, client selection cannot be manipulated.

\subsection{Cryptographic Primitives}
~\label{sec:design_primitives}

\system{} uses the following cryptographic primitives.

\PHM{Verifiable Random Function (VRF)~\cite{micali1999verifiable, dodis2005verifiable}.}
A VRF is a public-key PRF that provides proof that its outputs were computed correctly.
It consists of the following algorithms:

\begin{itemize}
	\item $(pk, sk) \xleftarrow{\$}$ \texttt{VRF.keygen}$(1^{\kappa})$.
	The key generation algorithm randomly samples a public/secret key pair $(pk, sk)$ from the security parameter $\kappa$.
	\item $(\beta, \pi) \leftarrow $\texttt{VRF.eval}$_{sk}(x)$.
	The evaluation function receives a binary string $x$ and a secret key $sk$ and generates a random string $\beta$ and the corresponding proof $\pi$.
	\item $\{0, 1\} \leftarrow $\texttt{VRF.ver}$(pk, x, \beta, \pi)$. The verification function receives a public key $pk$ and three binary strings $x$, $\beta$, and $\pi$, and outputs either 0 or 1.
\end{itemize}

A secure VRF should have unique provability and pseudorandomness.
The former requires that for every public key $pk$ and input $x$, there is a unique output $\beta$ for which a proof $\pi$ exists such that the verification function outputs 1.
The latter requires that no adversary can distinguish a VRF output without the accompanying proof from a random string.

\PHM{Signature Scheme.}
A signature scheme allows a message recipient to verify that the message came from a sender it knows.
It consists of the following algorithms:
\begin{itemize}
    \item $(pk, sk) \xleftarrow{\$}$ \texttt{SIG.gen}$(1^{\kappa}).$
    The key generation algorithm randomly samples a public/secret key pair $(pk, sk)$ from the security parameter $\kappa$.
    \item $\sigma \leftarrow $\texttt{SIG.sign}$_{sk}(m)$.
    The sign function receives a binary string $m$ and a secret key $sk$ and generates a signature $\sigma$.
    \item $\{0, 1\} \leftarrow $\texttt{SIG.ver}$(pk, m, \sigma)$. The verification function receives a public key $pk$ and two strings $m$ and $\sigma$, and outputs either 0 or 1.
\end{itemize}

For correctness, a signature scheme must ensure that the verification algorithm outputs 1 when presented with a correctly generated signature for any given message.
For security, we require that no PPT adversary, when given a fresh honestly generated public key and access to a signature oracle, can produce a valid signature on a message on which the oracle was queried with non-negligible probability.

\PHM{Pseudorandom Function (PRF)~\cite{goldreich1986construct}.}
A secure \texttt{PRF}$_k(x)$ is a family of deterministic functions indexed by a key $k$ that map an input $x$ into an output $y$ in such a way that $y$ is computationally indistinguishable from the output of a truly random function with input $x$.

\subsection{Secure Random Selection}
~\label{sec:design_random}

A faithful execution of random selection naturally achieves our goal (\cref{sec:design_problem}).
However, enforcing this in the presence of a malicious server poses two key challenges:

\begin{itemize}
    \item \emph{Correctness} (\textbf{C1}): A malicious server can deviate from random selection while falsely claiming adherence to the protocol.
    How can clients prevent such misbehavior?
    \item \emph{Consistency} (\textbf{C2}):
    Even if the protocol is faithfully executed, how can we ensure that the participants in the subsequent workflow are indeed those who were selected?
\end{itemize}

\PHM{Self-Sampling with Verifiable Randomness.}
To address \textbf{C1}, we propose a \emph{self-sampling} approach where each individual client in the population is responsible for determining whether it should participate in a training round based on publicly available inputs.
The intuition is that honest clients have a \emph{strong incentive} to follow a reasonable selection plan due to concerns about preserving their own privacy.
For security, we need to tackle the following issues simultaneously:

\begin{itemize}
	\item \emph{T1.1}: Ensuring that dishonest clients cannot participate without following the prescribed self-sampling protocol.
	\item \emph{T1.2}: Ensuring that honest clients who choose to participate do not get arbitrarily omitted by the malicious server.
\end{itemize}

To accomplish T1.1, \system{} introduce \emph{verifiable randomness} to the self-sampling process to allow each participant to verify the integrity of other peers' execution.
Specifically, we utilize a VRF with range $[0, m)$, and assume that the FL training needs to randomly sample $s$ out of $n$ clients for a specific training round $r$.
At the beginning of the round, the server first announces $s$, $n$, $r$ to all clients in the population.
Each client $i$ then generates a random number $\beta_i$ and the associated proof $\pi_i$ using its private key $sk_i^{\system{}}$ via $\beta_i, \pi_i = $ \texttt{VRF}$_{sk_i^{\system{}}}(r)$.
To achieve an expectation of $s$ sampled clients, only those clients with $\beta_i < sm/n$ will claim to participate by uploading their $\beta_i$'s and $\pi_i$'s to the server, while the remaining clients abort.
Next, the server takes all the surviving clients $P$ as participants and dispatches the collected $\{\beta_i, \pi_i\}_{i \in P}$ to each of them.
Finally, each participant $i$ verifies the validity and eligibility of its peers' randomness, i.e., checking if $\textrm{\texttt{VRF.val}}(pk_{j}^{\system{}}, r, \beta_j, \pi_j) = 1$ and $\beta_j < s m/n$ for each $j \in P$.
It aborts upon any failure, or refers to this set of participants throughout the remaining stages of the training round.

By pseudorandomness of VRFs (\cref{sec:design_primitives}), if each $sk_i^{\system{}}$ is uniformly random (guaranteed by the PKI as implied in~\cref{sec:design_threat}), this process is pseudorandom where each client is selected equally likely with probability $s/n$.
Also, the adversary cannot forge a random value that falls inside the range of $[0, sm/n)$ to help a dishonest client get selected.
By VRF's unique provability (\cref{sec:design_primitives}), if client $i$'s original randomness generated with $sk_i^{Lotto}$ and $r$ falls outside the range, no eligible proof can be produced for supporting any forged $\beta_i'$ to pass an honest client's verification test that takes $pk_i^{Lotto}$ and $r$ as inputs.

Regarding the round index $r$, it should be (i) unique for each round to avoid replay attacks, and (ii) consistently used by honest clients to preserve the sampling randomness.
\system{} achieves (i) by allowing each client to abort if the announced $r$ has been used previously, and also guarantees (ii) with a consistency check (as later detailed).
The above design, together with the key-pair integrity guaranteed by PKIs, ensures that each participant's randomness must be honestly generated.

\PHM{Over-Selection with Controlled Residual Removal.}
While VRFs securely achieve Bernoulli sampling, our primarily desired form of random selection is subsampling without replacement (\cref{sec:design_problem}).
To achieve a fixed sample size, \system{} further employs over-selection with residual removal based on the above process.
Specifically, instead of directly targeting $s$ sampled clients, \system{} slightly increases the expected sample size by a factor $\alpha>1$.
In this case, clients identify themselves as candidate participants $C$ when their associated $\beta_i$'s satisfy $\beta_i < \alpha sm/n$.
Next, the server samples $s$ candidates out of $C$ uniformly at random 
to finalize the set of participants $P$, if the total number of candidates is no less than the required sample size $s$ (the probability of this event is given by Theorem~\ref{thm:over_selection}) and aborts otherwise.
As we observed in practice (Figure~\ref{fig:over_selection}), $\alpha \geq 1.3$ suffices to sample enough clients with a high success rate (i.e., making Theorem~\ref{thm:over_selection}'s probability close to 1).
On receiving $P$, clients verify whether $\lvert P \rvert = s$ and abort if not.

\begin{figure}[t]
	\centering
	\includegraphics[width=0.9\columnwidth]{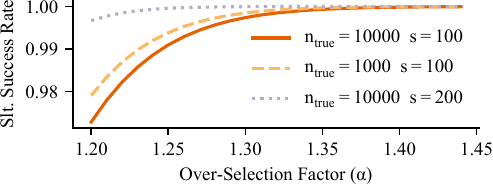}
	\caption{A small over-selection factor ($\alpha$) suffices in practice.}
	\label{fig:over_selection}
\end{figure}

\begin{theorem}[Effectiveness of Over-Selection]
	Using an over-selection factor $\alpha>0$, the described process in \system{} results in at least $s$ candidates with probability of $1 - \sum_{i = 0}^{s - 1} \binom{n_{true}}{i} p^i (1-p)^{n_{true} - i}$ where $p = \alpha s /n$ and $n_{true}$ and $n$ is the actual size of the population and that used by clients, respectively.
	\label{thm:over_selection}
\end{theorem}

In the above process, a malicious server can intentionally remove honest candidates, not just the ones with randomly chosen indices, to gain an advantage in growing a dishonest cohort.
Consider a scenario where there are $c$ dishonest clients among a population of size $n_{true}$, while the server announces a population size $n$ that may not necessarily align with $n_{true}$.
Suppose that the candidate set $C$ ($\lvert C \rvert \geq s$) contains $x$ dishonest clients after over-selection, then the expected proportion of them in $C$, $p_{cand}$, is equal to the base rate of those in the population, i.e., $p_{cand} = \mathbb{E} \left[ \frac{x}{\lvert C \rvert} \big\vert \lvert C \rvert \geq s \right] = c / n_{true}$.
Next, in the worst-case scenario where all the clients removed by the server are honest,  the expected proportion of dishonest clients in the final $s$ participants, $p_{final}$, is upper bounded as
\begin{equation*}
\begin{aligned}
	p_{final}
	&= \mathbb{E} \left[ \frac{x}{s} \bigg\vert \lvert C \rvert \geq s \right]
	= \frac{1}{s} \mathbb{E} \left[ x \big\vert \lvert C \rvert \geq s \right] \\
	&= \frac{1}{s} \cdot \frac{\mathbb{E} \left[ x \right] - \mathbb{E} \left[ x \big\vert \lvert C \rvert < s \right] \cdot \Pr \left[ \lvert C \rvert < s\right]}{\Pr  \left[ \lvert C \rvert \geq s\right]} \\
	&< \frac{1}{s} \cdot \frac{\mathbb{E} \left[ x \right]}{\Pr  \left[ \lvert C \rvert \geq s\right]}
	\approx \frac{1}{s} \mathbb{E} \left[ x \right] = \frac{1}{s} \cdot \frac{\alpha c s}{n} = \frac{\alpha c}{n}.
\end{aligned}
\end{equation*}
Thus, in expectation, the server can grow the proportion of dishonest clients by a factor approximately up to $\lambda = \frac{\alpha n_{true}}{n}$.

\begin{table*}[t]
	\footnotesize
	\begin{tabularx}{\linewidth}{X}
		\toprule
		\\
		\centerline{The \system{} Protocol for Secure Participant Selection} \\
		\vspace*{-0.2in}
		\begin{itemize}
			\item \textbf{Offline Setup}:
			\begin{itemize}
				\item[--]
				All parties are given the security parameter $\kappa$, the target number of sampled clients in a round $s$, the over-selection factor $\alpha$, a range $[0, m)$ to be used for \texttt{PRF} or \texttt{VRF},
				the population threshold $n_{min}$,
				and a timeout $l$.
				All clients also have a private authenticated channel with the server and the PKI, respectively.
				The server knows the exact population size $n$.
			\end{itemize}
			\(
			\textrm{\textit{Client} \ } i \left\{
			\begin{minipage}[c]{0.9\linewidth}
				\begin{itemize}
					\item[--]
					Generate a key pair $(pk_i^{reg}, sk_i^{reg}) \leftarrow \textrm{\texttt{KA.gen}}(pp)$ for signature and send the public key $pk_i^{reg}$ to the PKI for registration.
					\item[--]
					Receive the secret key for \system{}'s use, $sk_i^{\system{}}$, from the PKI.
					Send the public key $pk_i^{reg}$ to the server for registration. 
				\end{itemize}
			\end{minipage}
			\right.
			\)
			\(
			\textrm{\textit{Server} \ } \left\{
			\begin{minipage}[c]{0.9\linewidth}
				\begin{itemize}
					\item[--]
					Receive all clients' registration keys $\{pk_i^{reg}\}$.
					Use them to query the PKI on the application public keys $\{pk_i^{\system{}}\}$.
				\end{itemize}
			\end{minipage}
			\right.
			\)
			\item \informed{\textbf{Online Stage 0 (Population Refinement)}:} \newline
			\(
			\textrm{\textit{Server} \ } \left\{
			\begin{minipage}[c]{0.9\linewidth}
				\begin{itemize}
					\item[--] Based on specific criteria, select $n$ clients with the best utility from the population and take this set of clients as the new population.
				\end{itemize}
			\end{minipage}
			\right.
			\)
			\item \client{\textbf{Online Stage 1 (Client-Centric Selection)}}: \newline
			\(
			\textrm{\textit{Server} \ } \left\{
			\begin{minipage}[c]{0.9\linewidth}
				\begin{itemize}
					\item[--]
					Announce to all clients in the population the beginning of round $r$ and the population size $n$.
				\end{itemize}
			\end{minipage}
			\right.
			\)
			\(
			\textrm{\textit{Client} \ } i \left\{
			\begin{minipage}[c]{0.9\linewidth}
				\begin{itemize}
					\item[--]
					Upon receiving the population size $n$, verify that $n \geq n_{min}$. Abort if it fails.
					\item[--]
					Generate a random number $\beta_i, \pi_i = \textrm{\texttt{VRF.eval}}_{sk_i^{\system{}}}(r)$.
					Send to the server $(\beta_i, \pi_i)$ and claim to join only if $\beta_i < \alpha s m / n$.
				\end{itemize}
			\end{minipage}
			\right.
			\)
			\(
			\textrm{\textit{Server} \ } \left\{
			\begin{minipage}[c]{0.9\linewidth}
				\begin{itemize}
					\item[--]
					Insert responding clients into the candidate set $C$ and store their reported self-selection outcomes $\{(\beta_i, \pi_i)\}$ until timeout after $l$.
				\end{itemize}
			\end{minipage}
			\right.
			\)
			\item \server{\textbf{Online Stage 1 (Server-Centric Selection)}:} \newline
			\(
			\textrm{\textit{Server} \ } \left\{
			\begin{minipage}[c]{0.9\linewidth}
				\begin{itemize}
					\item[--]
					For each client $i$ in the population, compute $\beta_i = \textrm{\texttt{PRF}}_{pk_i^{\system{}}}(r)$ and insert the client into the candidate set $C$ only if $\beta_i < \alpha s m / n$.
				\end{itemize}
			\end{minipage}
			\right.
			\)
			\item \textbf{Online Stage 2 (Client Verification)}: \newline
			\(
			\textrm{\textit{Server} \ } \left\{
			\begin{minipage}[c]{0.9\linewidth}
				\begin{itemize}
					\item[--]
					If $\lvert C \rvert < s$, abort the current round $r$.
					Otherwise, select $s$ clients from $C$ uniformly at random to form the participant set $P \triangleq \{(i, pk_i^{reg}, \client{\beta_i, \pi_i)}\}$.
					Broadcast to all participants $i \in P$ this participant set $P$ \server{and the used population size $n$}.
				\end{itemize}
			\end{minipage}
			\right.
			\)
			\(
			\textrm{\textit{Client} \ } i \left\{
			\begin{minipage}[c]{0.9\linewidth}
				\begin{itemize}
					\item[--]
					Receive from the server a participant set $P_i$ and query the PKI on the public key $pk_j^{\system{}}$ for each participant $j \in P_i$ using $pk_j^{reg}$.
					\item[--]
					Verify whether $r$ has been used before, $(i, pk_i^{reg}, \client{\beta_i, \pi_i}) \in P_i$, $\lvert P_i \rvert = s$,
					\server{$n \geq n_{min}$},
					\client{$\beta_i < \alpha s m /n$, and  $\textrm{\texttt{VRF.val}}(pk_j, r, \beta_j, \pi_j) = 1$}
					\server{(or, $\textrm{\texttt{PRF}}_{pk_j^{\system{}}}(r) < \alpha s m / n$)} for $\forall j \in P_i$.
					\item[--]
					Abort if any of the test fails.
					Sign the observed participant list $P_i$ via $\sigma_i = \textrm{\texttt{SIG.sign}}_{sk_i^{reg}}(r \lvert \rvert P_i)$ and send the signature $\sigma_i$ to the server.
				\end{itemize}
			\end{minipage}
			\right.
			\)
			\item \textbf{Online Stage 3 (Consistency Check)}: \newline
			\(
			\textrm{\textit{Server} \ } \left\{
			\begin{minipage}[c]{0.9\linewidth}
				\begin{itemize}
					\item[--]
					Collect signatures from participants in $P' \subseteq P$ and abort in case of timeout after $l$.
					Broadcast $\{j, \sigma_j\}_{j \in P'}$  to each participant $i \in P'$.
				\end{itemize}
			\end{minipage}
			\right.
			\)
			\(
			\textrm{\textit{Client} \ } i \left\{
			\begin{minipage}[c]{0.9\linewidth}
				\begin{itemize}
					\item[--] Receive from the server the  set $\{j, \sigma_j\}_{j \in P'_i}$.
					\item[--] Verify $P'_i = P_i$ and  $\textrm{\texttt{SIG.ver}}(pk_j^{reg}, r \vert \vert P_i, \sigma_j) = 1$ for $\forall j \in P'_i$.  Abort on any failure, or refer to $P_i$ as participants at round $r$ hereafter.
				\end{itemize}
			\end{minipage}
			\right.
			\)
		\end{itemize} \\
		\bottomrule
		\captionof{figure}{
			A detailed description of the \client{Client-Centric} and \server{Server-Centric} protocol for secure participant selection in \system{} (\cref{sec:design}).
			\informed{Italicized parts inside square brackets are additionally required for informed selection.}
		}
		\label{fig:algo}
	\end{tabularx}
	\vspace*{-0.2in}
\end{table*}

To ensure that $\lambda$ is reasonably small (T1.2),
\system{} fixes the use of $\alpha$ to be a small value around 1.3 since it suffices to sample enough clients as above mentioned.
\system{} also ensures that $n$ is large enough through a propose-acknowledge process.
Specifically, the server proposes only once the value of $n$ to a client during its check-in.
The client then compares $n$ with its intended number $n_{min}$, i.e., the population size necessary to prevent substantial inflation of the dishonest cohort from its standpoint, and proceeds with the federation only if $n \geq n_{min}$.

\PHM{Inter-Client Consistency via Signature Scheme.}
While we have ensured that each participant $i$ receives a selection outcome where each member appears with upper-bounded probability, the issue of guaranteeing identical outcomes among all participants remains.
To address this concern, \system{} employs a secure signature scheme, which operates as follows:

\begin{itemize}
	\item Each client signs its received selection outcome $P_i$ with its registration secret key $sk_i^{reg}$, generating a signature $\sigma_i =$ \texttt{SIG.sign}$_{sk_i^{reg}}(r \vert \vert P_i)$ where $\vert \vert$ denotes concatenation.
	\item After collecting all signatures from the set of all responding clients $P'$, the server broadcasts to them the set $\{(j, \sigma_j)\}_{j \in P'}$ containing the received signatures.
	\item Upon receiving $\{(j, \sigma_j)\}_{j \in P'_i}$, each client $i$ verifies the correctness of all signatures and ensures that the corresponding participant lists are consistent with its observed one, i.e., $P'_i = P_i$ and \texttt{SIG.ver}$(pk_j^{reg}, r \vert \vert P_i, \sigma_j)$ = 1 for $\forall j \in P'_i$. 
	If any failure occurs during verification, the client aborts.
\end{itemize}

In summary, by addressing \textbf{C1} with the above design, honest participants are ensured to proceed only with a consistent participant list $P_i'$ of size $s$ with each member randomly presenting with probability less than $\alpha s / n_{min}$.
This upper bound remains \emph{invariant} during the training, as it is fully determined by public values $\alpha$, $s$, and $n_{min}$.
Hence, this bound can be used to derive practical security guarantees  (\cref{sec:design_security}), regardless of misbehaviors exhibited by the server.

\PHM{Inherent Inter-Stage Consistency with SecAgg.}
Ensuring consistency in participant sets used by the server throughout the entire FL workflow (\textbf{C2}) poses inherent challenges.
Indeed, if the privacy-preserving protocol used in the later stage mandates that all participants observe a genuine selection outcome to ensure valid computation, \textbf{C2} is immediately addressed.
\emph{However}, if the server can secretly involve a participant set different from the announced one without being noticed by honest clients, the attainability of \textbf{C2} becomes an open problem (\cref{sec:discuss}), as honest participants cannot detect such deception with limited trusted source of information (\cref{sec:design_threat}).

Fortunately, the protocol commonly employed in the remaining workflow is SecAgg (or the distributed differential privacy built upon it) (\cref{sec:background_defense}), where the server has an incentive to uphold the actual selection outcome during the entire training.
\emph{First}, it gains no advantage by secretly involving additional dishonest clients.
In SecAgg integrated with \system{}, an honest participant $i$ establishes secure channels with peer $j$ solely by using the public key $pk_j^{reg}$ contained in the selection outcome $P_i$ finalized by \system{}.
Therefore, it will not share any additional secrets with clients outside of $P_i$, even if they participate in the halfway.
\emph{Second}, the server will not pretend arbitrarily many dropout events of selected honest participants, e.g., by omitting their messages.
This is because SecAgg enforces a minimum number of participants (i.e., $t$) remaining in the final stage, beyond which the aggregation will abort without providing any meaningful information.

\begin{table}[t]
	\centering
	\caption{\system{}'s \textbf{c}om\textbf{p}utation and \textbf{c}om\textbf{m}unication complexity.}
	\label{tab:complexity}
	\resizebox{\columnwidth}{!}{%
		\begin{tabular}{@{}c|cccc|cccc@{}}
			\toprule
			\multirow{4}{*}{Variant} & \multicolumn{4}{c|}{Client-Centric}                                    & \multicolumn{4}{c}{Server-Centric}                                   \\ \cmidrule(l){2-9} 
			& \multicolumn{2}{c|}{Server}         & \multicolumn{2}{c|}{Participant} & \multicolumn{2}{c|}{Server}        & \multicolumn{2}{c}{Participant} \\ \cmidrule(l){2-9} 
			& cp  & \multicolumn{1}{c|}{cm} &  cp           & cm          & cp & \multicolumn{1}{c|}{cm} & cp          & cm          \\ \midrule
			Random &
			$O(1)$ &
			\multicolumn{1}{c|}{\multirow{2}{*}{$O(n+s^2)$}} &
			\multicolumn{2}{c|}{\multirow{2}{*}{$O(s)$}} &
			\multirow{2}{*}{$O(n)$} &
			\multicolumn{1}{c|}{\multirow{2}{*}{$O(s^2)$}} &
			\multicolumn{2}{c}{\multirow{2}{*}{$O(s)$}} \\
			Informed                 & $O(n)$ & \multicolumn{1}{c|}{}      & \multicolumn{2}{c|}{}            &       & \multicolumn{1}{c|}{}      & \multicolumn{2}{c}{}            \\ \bottomrule
		\end{tabular}%
	}
\end{table}

\PHM{Full Protocol: Client-Centric v.s. Server-Centric.}
Figure~\ref{fig:algo} summarizes \system{}'s overall design.
We refer to this algorithm as \texttt{Client-Centric} since the selection primarily occurs at the client end.
\texttt{Client-Centric} effectively achieves the security goal outlined in~\cref{sec:design_problem}.
We provide a detailed analysis of its security at~\cref{sec:design_security} and its complexity in Table~\ref{tab:complexity}.

It is worth noting that there exists an alternative design that may initially appear to be a natural choice.
This design retains the core elements of \texttt{Client-Centric} but relies on the server to perform the random selection.
Specifically, for each client $i$ in the population, the server uses a PRF with range $[0, m)$ to generate a random value $\beta_i = \textrm{\texttt{PRF}}_{pk_i^{\system{}}}(r)$, and selects it as a participant if $\beta_i < \alpha s m / n$.
The selection outcome, along with the associated randomness, is communicated to each participant for verification purposes.
We refer to this design as \texttt{Server-Centric}.
Intuitively, it offers the same security as \texttt{Client-Centric} does in a \emph{single training round}, as it also provides honest clients with a consistent participant list with members presenting with upper-bounded probability.

However, \texttt{Server-Centric} faces a unique security challenge during \emph{multiple-round training}.
As the selection outcome relies solely on public inputs, a malicious server can accurately predict, before online training, whether an honest participant $i$ will be selected in any round $r$.
This opens up opportunities for the adversary to exploit the system.
For instance, during residual removal after over-selection, the server can deliberately preserve the predicted most frequently appearing candidate.
In the context of DP training, this coordination behavior implies a weaker level of privacy.
We defer a rigorous assessment to Appendix~\ref{sec:appendix_centric}.
Considering the similar runtime cost of the two designs (\cref{sec:eval_efficiency}), we adopt \texttt{Client-Centric} as the first choice for its stronger security.

\subsection{Secure Informed Selection}
~\label{sec:design_informed}

We first note that the server, despite the potential of being malicious, is unlikely to tamper with the model utility.
This is because the server invests computing power in training and gains no advantage by doing otherwise.
If, however, it is indeed irrational, the need for informed selection becomes unnecessary, and we can safely fall back to using \system{}‘s secure random selection.
In both cases, \system{} aims to prevent a malicious server from manipulating the selection process for privacy attacks (\cref{sec:design_problem}).

\PHM{Challenges.}
Unlike random selection,  informed selection does not inherently incorporate randomness, but heavily relies on client measurements (\cref{sec:design_problem}).
This presents two challenges in our pursuit of controlling the number of dishonest participants.
Consider Oort~\cite{lai2021oort} as an example, where a client's response latency and training loss are employed as proxies for data quality and system speed, respectively (\cref{sec:design_problem}).
\emph{First}, the adversary can manipulate these metrics for selecting more colluding clients.
Lacking a ground truth view on these metrics, detecting such deception is hard for honest clients.
\emph{Second}, even if adversaries refrain from actively corrupting any measurements, they can still involve more dishonest clients by genuinely improving their data quality or system speed.
Such exploitation is more covert to detect.

\PHM{Approximation with Introduced Randomness.}
To side-step the above issues, \system{} introduces randomness to the selection process by incorporating secure random selection (\cref{sec:design_random}) to \emph{approximate} a given informed algorithm.
The server first \emph{refines} the population by excluding a fraction $d$ of clients with the lowest utility, based on the original criteria employed in the algorithm.
It then applies secure random selection to the refined population.
For example, if the original algorithm prioritizes the fastest clients, \system{} approximates it by first excluding the slowest $d$ of clients before random selection.

\PHM{Effectiveness of Population Refinement.}
Consider the following single-metric scenario without loss of generality: the original algorithm selects the best $s$ clients out of the initial population of size $n_{init}$ based on a metric $X$ (Case A).
When employing \system{}, $s$ clients are randomly selected from the refined population of size $n = n_{init}(1 - d)$ (Case B).

\begin{figure}[t]
	\centering
	\begin{subfigure}[b]{1.0\columnwidth}
		\centering
		\includegraphics[width=\columnwidth]{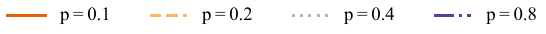}
		\label{fig:long_tailed_legend}
		\vspace{-1.07\baselineskip}
	\end{subfigure}
	\begin{subfigure}[b]{0.48\columnwidth}
		\centering
		\includegraphics[width=\columnwidth]{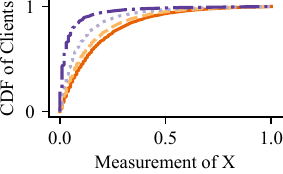}
		\caption{$X$ follows Zipf's distribution.}
		\label{fig:long_tailed_dist}
	\end{subfigure} \hfill
	\begin{subfigure}[b]{0.48\columnwidth}
		\centering
		\includegraphics[width=\columnwidth]{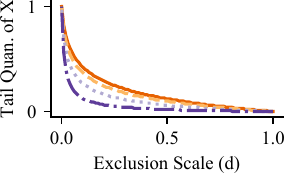}
		\caption{Effectiveness of refinement.}
		\label{fig:long_tailed_outcome}
	\end{subfigure}
	\caption{\system{} excels in approximating the original algorithm when the majority of the initial population is ``good''.}
	\label{fig:long_tailed}
\end{figure}

The extent to which \system{} approximates the original algorithm relies on the disparity in participant quality, as measured by $X$, between Case A and Case B.
This difference, in turn, depends on the distribution of $X$ across the initial population.
Intuitively, \system{} best approximates the original algorithm when $X$ follows a long-tailed distribution (e.g., power-law distributions~\cite{clauset2009power}), where the majority of the population is considered ``good" and only a small fraction is deemed ``bad''.
In this case, \system{} can exclude almost all the relatively few ``bad'' participants.
For example, suppose that $X$ is the smaller, the better, and it follows Zipf's distribution parameterized by $p$, such that the $X$ value of the $i$-th worst client is proportional to $i^{-p}$.
In this case, precluding the worst $d$ of the initial population enhances the quality of the worst client in the refined population by a factor of $(1-d)^{-p}$, as shown in Figure~\ref{fig:long_tailed}.

At the other extreme, there could be a scenario where the majority of clients are ``bad'' while only the minority are ``good.''
In this case, our approximation may not yield participants of quality comparable to the original algorithm.
We acknowledge this limitation yet contend that this compromise is a necessary cost in order to secure the selection process.

\PHM{Privacy-Utility Tradeoff under a Malicious Server.}
While excluding more ``bad'' clients helps create a better-refined population, it may also exclude more honest clients.
When the server behaves honestly, this does not alter the base rate of dishonest clients in the refined population $c/n$, since they are equally excluded.
However, a malicious server can grow this rate by intentionally excluding honest clients and/or more clients than planned.
To limit the impact of such misbehaviors, \system{} first enables the enforcement of a specific $d$ by having participants check the size of the refined population $n$ against the minimum acceptable one $n_{min}$ (\cref{sec:design_random}).

Next, the selection of $d$ involves a tradeoff between privacy and utility, where a smaller $d$ provides enhanced privacy but also reduces the potential utility.
We recommend utilizing the maximum permissible exclusion scale to attain a desired base rate of dishonest clients in the refined population, while minimizing the compromise on utility, as illustrated in Figure~\ref{fig:exclusion_scale}.
For instance, if $5\%$ of the initial client population consists of dishonest clients, using an exclusion scale no larger than 75\% can ensure that the base rate of dishonest clients in the refined population does not exceed $20\%$ in the worst case.

\begin{figure}[t]
	\centering
	\includegraphics[width=1.0\columnwidth]{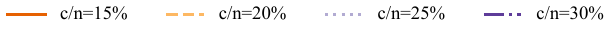}
	\includegraphics[width=1.0\columnwidth]{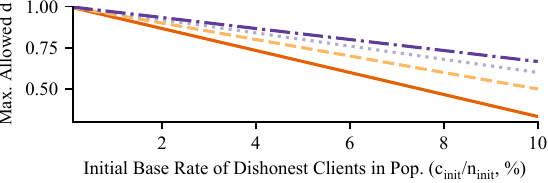}
	\caption{The maximum exclusion scale $d$ that allows to achieve a target base rate of dishonest clients in the refined population ($c/n$) given various initial base rates ($c_{init}/n_{init}$).}
	\label{fig:exclusion_scale}
\end{figure}

\PHM{Handling Multiple-Metric Algorithms.}
In the cases where the original algorithm incorporates two or more independent metrics, a straightforward approach to conducting population refinement is the \texttt{Or} strategy which excludes the worst clients based on one of the metrics.
Alternatively, we can also adopt the \texttt{And} strategy that excludes clients that perform poorly in all metrics simultaneously.
Furthermore, if the original informed algorithm combines the information from the two metrics into a single comprehensive metric, such as what Oort defines as the utility score (\cref{sec:design_problem}), there further exists a \texttt{Joint} strategy which directly relies on this single metric for population refinement.
We advocate the adoption of \texttt{Or} due to its superiority demonstrated in empirical experiments (\cref{sec:eval_approximation}).

\begin{figure*}[t]
	\centering
	\begin{subfigure}[b]{1.0\linewidth}
		\centering
		\includegraphics[width=1.0\linewidth]{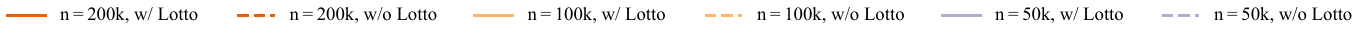}
	\end{subfigure} \newline
	\begin{subfigure}[b]{0.32\linewidth}
		\centering
		\includegraphics[width=\columnwidth]{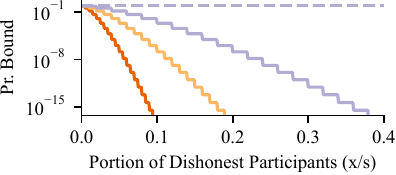}
		\caption{Direct implication.}
		\label{fig:thm_2_example}
	\end{subfigure} \hfill
	\begin{subfigure}[b]{0.32\linewidth}
		\centering
		\includegraphics[width=\columnwidth]{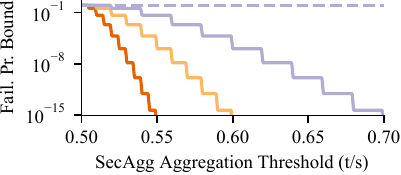}
		\caption{End-to-end impact with SecAgg.}
		\label{fig:cor_1_secagg}
	\end{subfigure} \hfill
	\begin{subfigure}[b]{0.32\linewidth}
		\centering
		\includegraphics[width=\columnwidth]{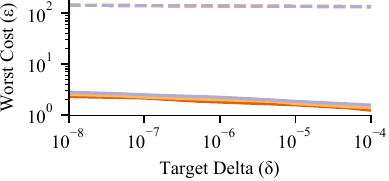}
		\caption{End-to-end impact with distributed DP.}
		\label{fig:cor_2_ddp}
	\end{subfigure}
	\caption{\system{}'s random selection effectively prevents arbitrary manipulation by the server in practical settings (\cref{sec:design_security}).}
	\label{fig:security_random}
\end{figure*}

\subsection{Security Analysis}
~\label{sec:design_security}

\PHB{Secure Random Selection.}
Our primary objective is to limit the number of dishonest participants (\cref{sec:design_problem}).
To assess how well \system{}'s random selection achieves it, we focus on the properties exhibited by the finalized participant list $P$.
Consider a population of size $n$, with $c$ clients being dishonest.
We denote by $s$ the desired number of sampled clients, and $x$ the number of clients that an adversary could manipulate among the selected participants.
Let $n_{min}$ be the population threshold, $m$ the randomness range, and $\alpha$ the over-selection factor.
Theorem~\ref{thm:security_random} holds despite the adversary's misbehavior.

\begin{theorem}[Security of Random Selection]
	If the protocol proceeds without abortion, for any $\eta > 1$, the probability that the proportion of dishonest participants, i.e., $x/s$, exceeds that in the population, i.e., $c/n$, is upper bounded as
		\begin{equation*}
		\Pr [\frac{x}{s} > \eta \frac{c}{n}] \leq 1 - \sum_{i=0}^{\lfloor \eta  c s / n \rfloor} \binom{c}{i}\left(\frac{1}{m} \lfloor \frac{\alpha sm}{n_{min}} \rfloor\right)^i \left(1 - \frac{1}{m} \lfloor \frac{\alpha sm}{n_{min}} \rfloor\right)^{c - i}.
		\end{equation*}
	\label{thm:security_random}
\end{theorem}

Intuitively, the larger the value of $\eta$, the smaller the right-hand side of the inequality mentioned above.
This clearly demonstrates the effectiveness of \system{} as a secure selection framework for achieving the desired property (\cref{sec:design_problem}).
While the proof is deferred to Appendix~\ref{sec:appendix_security}, the significance of Theorem~\ref{thm:security_random} in a practical scenario is visualized in Figure~\ref{fig:thm_2_example}.
Here, we vary the population size $n$, while the number of colluding clients in the population set as $c = 10^3$, and the target number of participants as $s = n/1000$.
We consider $\alpha=1.3$, $m=2^{256}$, and $n_{min} = n$.
When the population size $n=2\times10^5$, employing insecure random selection results in a probability of 1 for the event where the fraction of dishonest participants exceeds 5\% (i.e., $\eta c/n$ with $\eta = 10$).
However, with \system{}, this probability becomes negligible, approximately $1.3 \times 10^{-7}$.

We further use Theorem~\ref{thm:security_random} to analyze the security of SecAgg and distributed DP (\cref{sec:background_defense}) when protected by \system{}.
The related proof can also be found in Appendix \ref{sec:appendix_security}.
 
 \begin{corollary}[Bounded Failure Probability in SecAgg~\cite{bonawitz2017practical}]
 	Let $t$ be the aggregation threshold of SecAgg (\cref{sec:background_defense}).
 	With \system{}'s random selection, the probability of the server being able to observe an individual update of an honest client, which indicates a failure of SecAgg, is upper bounded as
 	\begin{equation*}
 		\Pr [Fail] \leq 1 - \sum_{i=0}^{2t- s - 1} \binom{c}{i}\left(\frac{1}{m} \lfloor \frac{\alpha sm}{n_{min}} \rfloor\right)^i \left(1 - \frac{1}{m} \lfloor \frac{\alpha sm}{n_{min}} \rfloor\right)^{c - i}.
 	\end{equation*}
 	\label{cor:security_secagg}
 \end{corollary}
 
\begin{corollary}[Controlled Privacy Cost in Distributed DP, I]
 	Let $\pi$ be the distributed DP protocol parameterized by $\sigma$ and built atop SecAgg with aggregation threshold $t$.
 	Given target $\delta$, an $R$-round FL training with $\pi$ and \system{}'s  random selection (\texttt{Client-Centric}) achieves $(\epsilon, \delta)$-DP, where
 	\begin{equation*}
 		\begin{aligned}
 			\epsilon &= \min_{\substack{0 \leq k \leq \min \{c, 2t - s - 1\} \\ 0 \leq r \leq R \\ p_k + q_r < \delta}} E_\pi(s, k, r, \sigma, 1 - \frac{1-\delta}{1 - p_k - q_r}),\\
 			p_k &= 1 - \left( \sum_{i=0}^k \binom{c}{i} \left(\frac{1}{m} \lfloor \frac{\alpha s m}{n_{min}} \rfloor \right)^i \left(1 - \frac{1}{m}\lfloor \frac{\alpha s m}{n_{min}} \rfloor \right)^{c - i} \right)^R,
 		\end{aligned}
 	\end{equation*}
 	with $\sigma$ the noise multiplier used in $\pi$, and $E_\pi(\cdot)$ the privacy accounting method of $\pi$.
 	In addition, $q_r = \phi_{R, r}$ with the recurrence definition of $\phi$ being $\phi_{j, r} = 1 - \gamma(j, r) + \gamma(j, r) \phi_{j - 1, r}$ with boundary condition $\phi_{r, r} = 0$ and $\gamma(j, r) = $
 	\begin{equation*}
 		\left( \sum_{i=0}^{r - 1} \binom{j - 1}{i} \left(\frac{1}{m} \lfloor \frac{\alpha s m}{n_{min}} \rfloor \right)^i \left(1 - \frac{1}{m}\lfloor \frac{\alpha s m}{n_{min}} \rfloor \right)^{j - 1 - i} \right)^s.
 	\end{equation*}
	\label{cor:security_ddp}
\end{corollary}

Corollary~\ref{cor:security_secagg} provides insights into the robustness of \system{} in maintaining the security of SecAgg, while Corollary~\ref{cor:security_ddp} outlines the end-to-end privacy guarantees offered by \system{} when integrated with existing distributed DP systems.
In the specific case discussed above, when SecAgg is used with $t$ set larger than $0.53s$\footnote{The feasible range of $t$  SecAgg is $(0.5s, s]$ in the malicious settings~\cite{bonawitz2017practical}.}, the probability of SecAgg's failure diminishes to less than $8.9 \times 10^{-6}$ when employing \system{}. 
Otherwise, the probability is 1 (Figure~\ref{fig:cor_1_secagg}).
Building upon SecAgg with $t = 0.7s$, if DSkellam~\cite{agarwal2021skellam} is implemented as the distributed DP protocol, with a target $\delta = 1/n$ and using a consistent noise multiplier, training over FEMNIST (refer to the related settings in~\cref{sec:eval_method}), \system{} guarantees a moderate privacy cost of $\epsilon=1.8$. In contrast, insecure selection would result in a prohibitive cost of $\epsilon=265$ (Figure~\ref{fig:cor_2_ddp}).

\PHM{Secure Informed Selection.}
Population refinement in \system{}'s informed selection can be deemed as a preprocessing step that precedes its random selection (\cref{sec:design_informed}).
Hence, the security guarantee provided can be inherited from the random selection.
Assume that the initial population consists of $n_{init}$ clients, with $c$ being the number of dishonest clients.
Denote by $n$ the size of the refined population.
As the server can act maliciously by deliberately excluding honest clients, in the worst case, all $c$ dishonest clients remain in the refined population.
Therefore, the security guarantee can be derived by \emph{directly} applying Theorem~\ref{thm:security_random}.
\section{Implementation}~\label{sec:impl}

We have implemented \system{} as a library easily pluggable into existing FL systems with 1235 lines of Python code.
It operates within the server and on each client.

\begin{figure}[t]
	\centering
	\includegraphics[width=1.0\columnwidth]{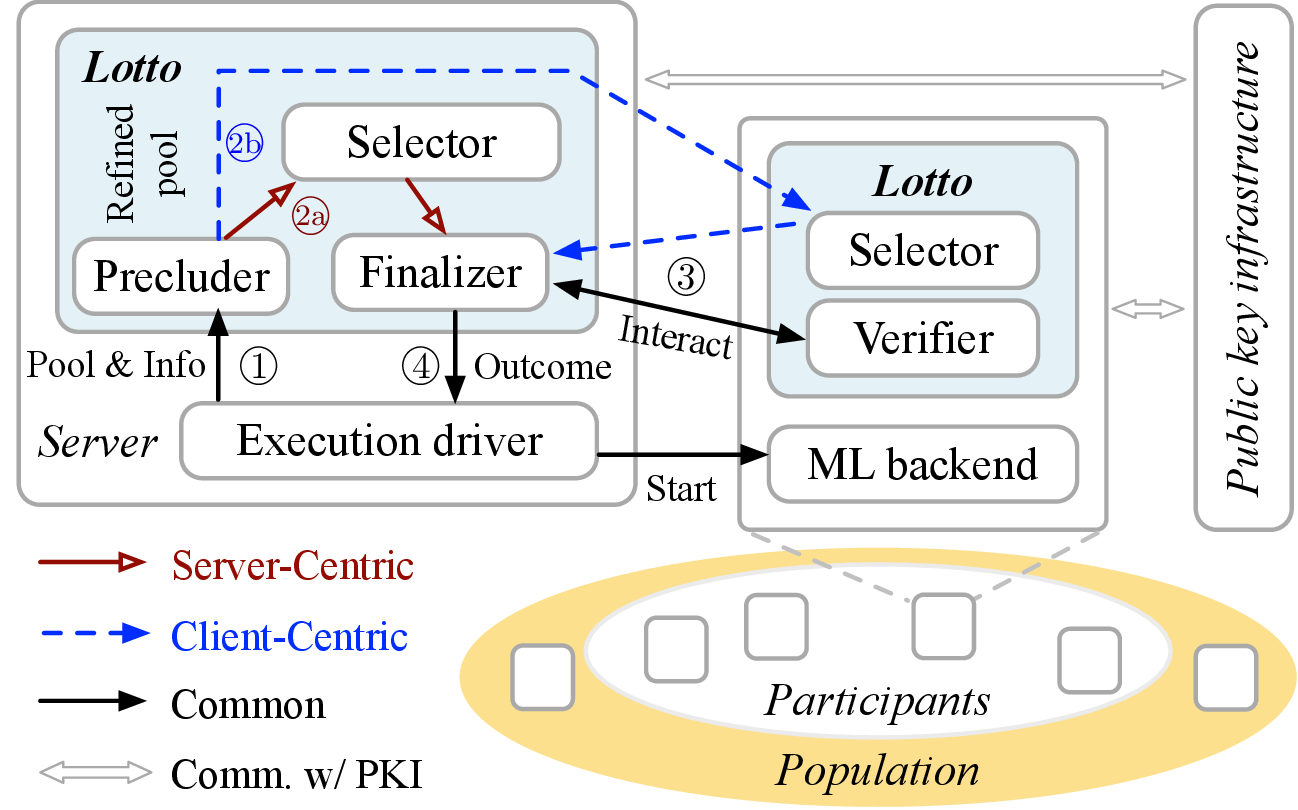}
	\caption{\system{} architecture.}
	\label{fig:arch}
\end{figure}

\PHM{System Workflow.}
Figure~\ref{fig:arch} shows how \system{} interacts with the FL execution framework during participant selection.
\textcircled{\raisebox{-1.0pt}{1}}
\emph{Pool Confirmation}: the framework forms a pool with clients meeting eligibility properties (e.g., battery level, \cref{sec:background_fl}) and forward their identities and characteristics (e.g., latency) to \system{}.
In case of informed selection, the server end refines the pool by excluding worst-performing clients based on specific criteria.
\system{} next proceeds to random selection and has two variants.
\textcircled{\raisebox{-0.4pt}{\footnotesize 2a}}
\emph{Server-Centric Selection}:  the server end selects participants from the pool as candidates using PRFs.
\textcircled{\raisebox{-0.4pt}{\footnotesize 2b}}
\emph{Client-Centric Selection}: for each client in the pool, the server end informs its client end to self-sample with VRFs.
Clients meeting the selection conditions report their results and become candidates.
\textcircled{\raisebox{-1.0pt}{3}} \emph{Mutual Verification}: the server selects the target number of clients from candidates as participants.
It then helps each participant verify the participant list.
\textcircled{\raisebox{-1.0pt}{4}} \emph{Finalization}:
The FL framework proceeds with the participant list output by \system{} and starts FL training.

\PHM{Cryptographic Instantiation.}
Secure PRFs can be implemented with a cryptographic hash function in conjunction with a secret key, and we use \texttt{HMAC-SHA-256}.
For secure signature schemes, we utilize \texttt{Ed25519}~\cite{bernstein2012high}, an elliptic-curve-based signature scheme.
Both of the above choices are implemented leveraging the \texttt{cryptography} Python library~\cite{cryptography}.
We implement VRFs using \texttt{ECVRF-Ed25519-SHA512-Elligator2} as specified in~\cite{draft} using the codebase provided by~\cite{vrf}.

\section{Evaluation}~\label{sec:eval}

\begin{table*}[h]
	\centering
	\caption{Per-round training time and network transfer cost for the server and average participating client with random selection.}
	\label{tab:rand_perf}
	\resizebox{\textwidth}{!}{%
		\begin{tabular}{@{}cccccccccccccc@{}}
			\toprule
			\multicolumn{2}{c|}{FL Application} &
			\multicolumn{4}{c|}{FEMNIST@CNN} &
			\multicolumn{4}{c|}{OpenImage@MobileNet} &
			\multicolumn{4}{c}{Reddit@Albert} \\ \midrule
			\multicolumn{1}{c|}{\multirow{2.5}{*}{Population}} &
			\multicolumn{1}{c|}{\multirow{2.5}{*}{Protocol}} &
			\multicolumn{2}{c|}{Time} &
			\multicolumn{2}{c|}{Network} &
			\multicolumn{2}{c|}{Time} &
			\multicolumn{2}{c|}{Network} &
			\multicolumn{2}{c|}{Time} &
			\multicolumn{2}{c}{Network} \\ \cmidrule(l){3-14} 
			\multicolumn{1}{c|}{} &
			\multicolumn{1}{c|}{} &
			Server &
			\multicolumn{1}{c|}{Client} &
			Server &
			\multicolumn{1}{c|}{Client} &
			Server &
			\multicolumn{1}{c|}{Client} &
			Server &
			\multicolumn{1}{c|}{Client} &
			Server &
			\multicolumn{1}{c|}{Client} &
			Server &
			Client \\ \midrule
			\multicolumn{1}{c|}{\multirow{3}{*}{100}} &
			\multicolumn{1}{c|}{Rand} &
			1.76min &
			\multicolumn{1}{c|}{0.97min} &
			64.88MB &
			\multicolumn{1}{c|}{3.9MB} &
			3.06min &
			\multicolumn{1}{c|}{2.28min} &
			64.35MB &
			\multicolumn{1}{c|}{3.87MB} &
			13.0min &
			\multicolumn{1}{c|}{6.67min} &
			958.55MB & 57.46MB
			\\
			\multicolumn{1}{c|}{} &
			\multicolumn{1}{c|}{Cli-Ctr} &
			1.86min &
			\multicolumn{1}{c|}{1.26min} &
			64.94MB &
			\multicolumn{1}{c|}{3.9MB} &
			3.07min &
			\multicolumn{1}{c|}{2.44min} &
			64.4MB &
			\multicolumn{1}{c|}{3.87MB} &
			12.86min &
			\multicolumn{1}{c|}{8.8min} &
			958.6MB & 57.46MB
			\\
			\multicolumn{1}{c|}{} &
			\multicolumn{1}{c|}{Srv-Ctr} &
			1.77min &
			\multicolumn{1}{c|}{0.97min} &
			64.89MB &
			\multicolumn{1}{c|}{3.9MB} &
			2.97min &
			\multicolumn{1}{c|}{2.17min} &
			64.36MB &
			\multicolumn{1}{c|}{3.87MB} &
			12.88min &
			\multicolumn{1}{c|}{6.58min} &
			958.86MB & 57.46MB
			\\ \midrule
			\multicolumn{1}{c|}{\multirow{3}{*}{400}} &
			\multicolumn{1}{c|}{Rand} &
			2.56min &
			\multicolumn{1}{c|}{1.4min} &
			0.26GB &
			\multicolumn{1}{c|}{3.56MB} &
			4.35min &
			\multicolumn{1}{c|}{3.36min} &
			0.25GB &
			\multicolumn{1}{c|}{3.53MB} &
			26.94min &
			\multicolumn{1}{c|}{15.65min} &
			3.75GB & 51.53MB
			\\
			\multicolumn{1}{c|}{} &
			\multicolumn{1}{c|}{Cli-Ctr} &
			2.59min &
			\multicolumn{1}{c|}{1.83min} &
			0.26GB &
			\multicolumn{1}{c|}{3.56MB} &
			4.68min &
			\multicolumn{1}{c|}{3.89min} &
			0.25GB &
			\multicolumn{1}{c|}{3.53MB} &
			27.53min &
			\multicolumn{1}{c|}{21.95min} &
			3.75GB & 51.53MB
			\\
			\multicolumn{1}{c|}{} &
			\multicolumn{1}{c|}{Srv-Ctr} &
			2.29min &
			\multicolumn{1}{c|}{1.3min} &
			0.26GB &
			\multicolumn{1}{c|}{3.56MB} &
			4.51min &
			\multicolumn{1}{c|}{3.49min} &
			0.25GB &
			\multicolumn{1}{c|}{3.53MB} &
			27.17min &
			\multicolumn{1}{c|}{15.76min} &
			3.75GB & 51.53MB
			\\ \midrule
			\multicolumn{1}{c|}{\multirow{3}{*}{700}} &
			\multicolumn{1}{c|}{Rand} &
			3.46min &
			\multicolumn{1}{c|}{2.01min} &
			0.45GB &
			\multicolumn{1}{c|}{3.69MB} &
			5.65min &
			\multicolumn{1}{c|}{4.1min} &
			0.45GB &
			\multicolumn{1}{c|}{3.66MB} &
			40.06min &
			\multicolumn{1}{c|}{24.77min} &
			6.56GB & 52.57MB
			\\
			\multicolumn{1}{c|}{} &
			\multicolumn{1}{c|}{Cli-Ctr} &
			3.82min &
			\multicolumn{1}{c|}{2.82min} &
			0.45GB &
			\multicolumn{1}{c|}{3.69MB} &
			6.23min &
			\multicolumn{1}{c|}{5.06min} &
			0.45GB &
			\multicolumn{1}{c|}{3.66MB} &
			39.59min &
			\multicolumn{1}{c|}{33.91min} &
			6.56GB & 52.57MB
			\\
			\multicolumn{1}{c|}{} &
			\multicolumn{1}{c|}{Srv-Ctr} &
			3.56min &
			\multicolumn{1}{c|}{2.02min} &
			0.45GB &
			\multicolumn{1}{c|}{3.7MB} &
			5.62min &
			\multicolumn{1}{c|}{4.06min} &
			0.45GB &
			\multicolumn{1}{c|}{3.66MB} &
			38.85min &
			\multicolumn{1}{c|}{23.84min} &
			6.56GB & 52.57MB
			\\ \bottomrule
		\end{tabular}%
	}
\end{table*}

We evaluate \system{}’s effectiveness in various FL training testbeds.
The highlights of our evaluation are listed below.

\begin{enumerate}
	\item
	Compared to insecure selection methods, \system{} introduces at most 10\% runtime overhead and negligible extra network usage across different participation scales (\cref{sec:eval_efficiency}).
	\item In the case of informed selection, the adoption of \system{} achieves comparable or even superior training efficiency compared to that of insecure methods (\cref{sec:eval_approximation}).
\end{enumerate}

\subsection{Methodology}
~\label{sec:eval_method}

\PHB{Datasets and Models.}
We run two common categories of applications with three real-world datasets of different scales.

\begin{itemize}
	\item \emph{Image Classification}:
	the FEMNIST~\cite{caldas2018leaf} dataset, with 805k images classified into 62 classes, and a more complicated dataset OpenImage~\cite{openimage} with 1.5M images spanning 600 categories.
	We train a CNN~\cite{kairouz2021distributed, agarwal2021skellam}, and MobileNet V2 ~\cite{sandler2018mobilenetv2} with 0.25 as width multiplier to classify the images in FEMNIST and OpenImage, respectively.
	\item \emph{Language Modeling}: the large-scale Reddit dataset~\cite{reddit}.
	We train the Albert model~\cite{lan2020albert} for next-word prediction.
\end{itemize}

\PHM{Cluster Setup.}
We launch an AWS EC2 \texttt{r5n.8xlarge} instance (32 vCPUs, 256 GB memory, and 25 Gbps network) for the server and one \texttt{c5.xlarge} (4 vCPUs and 8 GB memory) instance for each client, aiming to match the computing power of mobile devices.
We also throttle clients’ bandwidth to fall at 44Mbps to match the average mobile bandwidth~\cite{ciscoannual}.

\PHM{Hyperparameters.}
For local training, we use the mini-batch SGD for FEMNIST and OpenImage and AdamW~\cite{loshchilov2018decoupled} for Reddit, all with momentum set to 0.9.
The number of training rounds and local epochs are 50 and 2 for both FEMNIST and Reddit and 100 and 3 for OpenImage, respectively.
The batch size is consistently 20, while the learning rate is fixed at 0.01 for FEMNIST and 8e-5 for Reddit, respectively.
For OpenImage, the initial learning rate is $0.05$, and decreases by a factor of 0.98 after every 10 rounds.
For model aggregation, we apply weighted averaging~\cite{mcmahan2017communication} protected by the SecAgg~\cite{bonawitz2017practical} protocol (\cref{sec:background_defense}).

\PHM{Baselines.}
In evaluating random selection,  we compare \system{} against Rand~\cite{mcmahan2017communication}, the naive random selection protocol.
As for informed selection, we compare \system{} against Oort~\cite{lai2021oort}, the state-of-the-art informed algorithm designed for maximum training efficiency.
Both Oort and Rand are insecure as the server has the freedom to include dishonest clients.
For \system{}, we primarily evaluate \texttt{Client-Centric}, while also demonstrating its similarity to \texttt{Server-Centric}  (\cref{sec:design_random}).

\subsection{Execution Efficiency}
\label{sec:eval_efficiency}

We assess the time performance and network overhead of \system{}.
To study the impact of the federation scale, we vary the population size from 100 to 400 to 700.
The target sampled clients are proportionately set to 10, 40, and 70, respectively.

\PHM{\system{} Induces Acceptable Overhead in Time.}
Table~\ref{tab:rand_perf} reports the average round time measured at the server when random selection is deployed.
Compared to Rand, training with \texttt{Client-Centric} incurs a time cost of less than 6\%, 8\%, and 10\% for population sizes of 100, 400, and 700, respectively.
Moreover, the overhead for \texttt{Server-Centric} is 1\%, 4\%, and 3\%, respectively.
Such acceptable overhead is expected as \system{} is built on lightweight security primitives (\ref{sec:design_random}).
Table~\ref{tab:rand_perf} also reports the time cost from the perspective of participants.
In comparison to Rand, \texttt{Server-Centric} incurs less than 4\% overhead at the participant side, again; whereas \texttt{Client-Centric} incurs an overhead of up to 40\% due to the offloading of the client selection process to each client.\footnote{Such inflation barely affects the end-to-end latency (measured at the server end) as participants run in parallel with the server.}
Regarding informed selection, the observed trends are similar to those reported in Table~\ref{tab:rand_perf} (with Rand replaced by Oort) and are deferred to Appendix~\ref{sec:appendix_informed}.

\PHM{\system{} Induces Negligible Overhead in Network.}
Table~\ref{tab:rand_perf} also provides the network footprint per round, measured at both the server end and the participant end.
As one can see, \system{} consistently introduces less than 1\% network overhead across all evaluated tasks and participation scales.
This can be attributed to the fact that the predominant communication overhead in FL is the transmission of the model, rather than other auxiliary information.
For example, when training over Reddit with 700 clients, model transfer in a round yields a network footprint of approximately 6.56 GB to the server.
In contrast, the extra cost brought by \texttt{Client-Centric} and \texttt{Server-Centric} are merely 1.3 MB and 0.3 MB, respectively.
As for informed selection, again, we defer the results to Appendix~\ref{sec:appendix_informed} given the similar implications.

\PHM{Client-Centric v.s. Server-Centric.}
The preceding analysis emphasizes \system{}'s minimal time and network costs, irrespective of whether it is implemented using a \texttt{Client-Centric} or \texttt{Server-Centric} approach.
Given their comparable efficiency, we recommend adopting \texttt{Client-Centric}, as it ensures enhanced security (\cref{sec:design_random}).

\begin{figure}[t]
	\centering
	\begin{subfigure}[b]{0.48\columnwidth}
		\centering
		\includegraphics[width=\columnwidth]{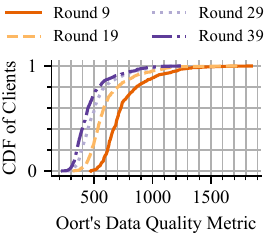}
		\caption{FEMNIST data.}
		\label{fig:hetero-femnist}
	\end{subfigure} \hfill
	\begin{subfigure}[b]{0.48\columnwidth}
		\centering
		\includegraphics[width=\columnwidth]{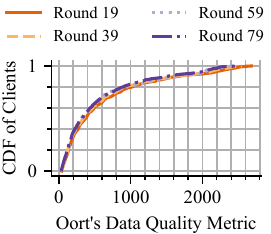}
		\caption{OpenImage data.}
		\label{fig:hetero-openimage}
	\end{subfigure}
	\begin{subfigure}[b]{0.48\columnwidth}
		\centering
		\includegraphics[width=\columnwidth]{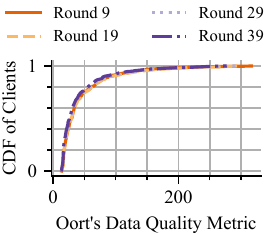}
		\caption{Reddit data.}
		\label{fig:hetero-reddit}
	\end{subfigure} \hfill
	\begin{subfigure}[b]{0.48\columnwidth}
		\centering
		\includegraphics[width=\columnwidth]{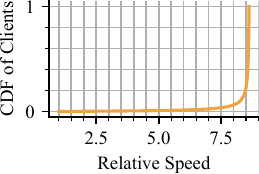}
		\caption{Client system speed.}
		\label{fig:hetero-speed}
	\end{subfigure}
	\caption{Visualization of the heterogeneous settings.}
	\label{fig:hetero}
\end{figure}

\begin{figure*}[t]
	\centering
	\begin{subfigure}[b]{0.33\linewidth}
		\centering
		\includegraphics[width=\columnwidth]{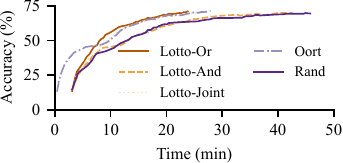}
		\caption{FEMNIST@CNN}		\label{fig:training-femnist}
	\end{subfigure} \hfill
	\begin{subfigure}[b]{0.33\linewidth}
		\centering
		\includegraphics[width=\columnwidth]{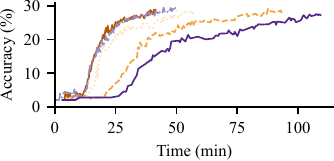}
		\caption{OpenImage@MobileNet.}
		\label{fig:training-openimage}
	\end{subfigure} \hfill
	\begin{subfigure}[b]{0.33\linewidth}
		\centering
		\includegraphics[width=\columnwidth]{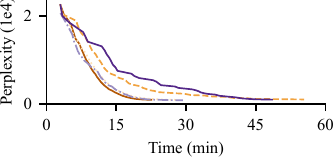}
		\caption{Reddit@Albert.}
		\label{fig:training-reddit}
	\end{subfigure}
	\caption{\system{} (with \texttt{Or} strategy) effectively approximates Oort and significantly outperforms Rand in training efficiency.}
	\label{fig:training}
\end{figure*}

\subsection{Effectiveness of Informed Selection}
\label{sec:eval_approximation}

We next evaluate the effectiveness of \system{}'s approximation strategies in achieving the objective of the target informed selection algorithm.
We focus on the \texttt{Client-Centric} implementation of \system{} and use the 700-client testbed (\cref{sec:eval_efficiency}).

\PHM{Heterogeneous Environment.}
The realistic datasets we employ (\cref{sec:eval_method}) already demonstrate data heterogeneity, as they are partitioned by the original data owners.
Figure~\ref{fig:hetero-femnist} to~\ref{fig:hetero-reddit} provides a visualization of the distribution of clients' data quality, as measured by Oort's definition (the first term in Equation~\eqref{eq:oort_score}) without loss of generality.
Snapshots are taken across different rounds of the training.
For Oort's considerations on system speed to be relevant, we additionally emulate hardware heterogeneity.
To this end, we set the response latencies of clients to consistently follow a Zipf distribution parameterized by $a = 1.2$ (moderately skewed), regardless of the specific FL task, as illustrated in Figure \ref{fig:hetero-speed}.

\PHM{The \texttt{Or} Strategy Outperforms Other Choices of \system{}.}
\system{} considers various strategies for population refinement when approximating the target informed algorithm (\cref{sec:design_informed}).
Here, we compare them using a consistent exclusion threshold of 20\%.
Under this threshold, a client is excluded if it ranks in the bottom 20\% in either system speed or data quality (\texttt{Or}); in both metrics (\texttt{And}); or in the comprehensive utility score (Equation~\eqref{eq:oort_score}) defined by Oort (\texttt{Joint}).
Additionally, we include the Rand method as a baseline reference.
As illustrated in Figure~\ref{fig:training}, \texttt{Or} and \texttt{Joint} constantly outperform Rand, while \texttt{And} fails to surpass it in certain cases, e.g., FEMNIST (Fig.~\ref{fig:training-femnist}).
This can be attributed to the independent distributions of system speed and data quality across clients.
As there are only a few clients that rank in the bottom 20\% simultaneously for both metrics (observed to be $\leq$5\% of the population), the number of clients that \texttt{And} excludes is inadequate.
Between \texttt{Or} and \texttt{Joint}, we advocate the use of the former due to its consistently superior performance and fix its use for the remaining analysis.

\PHM{\system{} Achieves Comparable or Even Superior Performance to Oort.}
To examine the time-to-accuracy performance, we set the target accuracy at 69.8\% and 27.6\% for FEMNIST and OpenImage, respectively.
For Reddit, we set the target perplexity (lower values are better) at 1010.
As depicted in Figure~\ref{fig:training}, \system{} achieves nearly identical performance to Oort in OpenImage, and it even surpasses Oort by 1.1$\times$ for both FEMNIST and Reddit.
In comparison to Rand, \system{} outperforms it by 2.1$\times$, 3.0$\times$, and 1.6$\times$ in FEMNIST, OpenImage, and Reddit, respectively.
Note that we have accounted for the disparity in execution efficiency between \system{} and baselines (\cref{sec:eval_efficiency}) for fairness of comparison.
These results underscore the significant enhancement in the quality of selected clients through population refinement, specifically employing the \texttt{Or} strategy with a 20\% exclusion factor, making the following random selection competitive with Oort which cherrypicks the best clients.
Overall, the randomness introduced by \system{} for enhanced security does not compromise the training efficiency achieved by insecure methods.

\section{Discussion and Future Work}~\label{sec:discuss}

\PHB{Enhanced Support for SecAgg and Distributed DP.}
\system{} achieves a proportion of compromised participants close to the base rate of dishonest clients in the population, regardless of its magnitude (Theorem~\ref{thm:security_random}).
Following this property, given an honest majority in the population, \system{} further ensures an honest majority among the participants.
This entertains both secure aggregation (by tightening the bound in Corollary~\ref{cor:security_secagg}) and distributed differential privacy (by lowering the bound in Corollary~\ref{cor:security_ddp}).
This condition can hold in practice, as simulating or compromising a large number of clients incurs prohibitive costs for the adversary given a vast client population in FL.\footnote{If there are only tens or hundreds of clients, participant selection is no longer needed as all clients can be involved simultaneously~\cite{kairouz2019advances}.}
Specifically, it is expensive to register millions of identities to a PKI, which usually requires a unique, verifiable identifier for validation~\cite{bonawitz2017practical, bell2020secure}.
For example, applying for Hong Kong Post eCert for individual PKI use mandates personal IDs~\cite{ecert}.
Moreover, the costs involved in operating a client botnet at scale are also excessively high~\cite{shejwalkar2022back}.

\PHM{General Solutions for Inter-Stage Consistency.}
\system{} does not explicitly ensure consistent reference to a participant list throughout the entire FL workflow.
This is because, for SecAgg this paper primarily focuses on, any inconsistent reference by the server offers no significant advantage in probing clients' data (\cref{sec:design_random}).
This applies to other secure aggregation protocols, such as SecAgg+~\cite{bell2020secure}, the state-of-the-art follow-up work of SecAgg.
Similar to SecAgg, SecAgg+ ensures that each participant will not share secrets with clients beyond its verified participant list, and also includes controllable dropout tolerance that can prevent the server from pretending the dropout of numerous honest participants.
However, it would be beneficial to explore general solutions for achieving inter-stage consistency when combining \system{} with other potential privacy-enhancing techniques in the future.

\section{Related Work}~\label{sec:related}

\PHB{Privacy-Preserving Aggregation with Malicious Server.}
Numerous protocols have been proposed to enable a server to learn the sum of multiple inputs without gaining access to individual inputs.
However, many of these protocols assume an honest-but-curious server~\cite{kadhe2020fastsecagg, so2021turbo, so2022lightsecagg} or a trusted third party~\cite{bittau2017prochlo, cheu2019distributed, erlingsson2019amplification}, which has limited practicality.
Only the commonly used SecAgg~\cite{bonawitz2017practical} and SecAgg+~\cite{bell2020secure} and a few other protocols~\cite{ma2023flamingo, liu2023dhsa, rathee2022elsa} allow for a malicious server.
To further control the information leakage through the sum,  distributed DP protocols~\cite{kairouz2021distributed, agarwal2021skellam, stevens2022efficient, jiang2024efficient} provide rigorous privacy guarantees by combining the DP noise addition with SecAgg.
However, both secure aggregation and distributed DP are still vulnerable to the privacy risks posed by a dishonest majority among participants, while \system{} takes the initiative to tackle this issue by secure participant selection.

\PHM{Secure Model Utility against Adversaries.}
Besides probing data privacy, an adversary may also compromise the global model performance~\cite{ bagdasaryan2020backdoor, fang2020local, shejwalkar2021manipulating, wang2020attack}.
Several approaches have been proposed to mitigate these attacks, e.g., by employing certifiable techniques during local training~\cite{cao2021provably, xie2021crfl,   burkhalter2021rofl} and model aggregation.
These approaches back up \system{} with robust model utility, whereas \system{} enhances data privacy.

\section{Conclusion}~\label{sec:conclusion}

Several aggregation protocols have been developed in FL to enhance the privacy of clients' data. However, the participant selection mechanisms used in these protocols are vulnerable to manipulation by adversarial servers. We proposed \system{} to address this problem. \system{} offers two secure selection algorithms, random and informed, both of which effectively align the fraction of compromised participants with the base rate of dishonest clients in the population. Theoretical analysis and large-scale experiments demonstrate that \system{} not only provides enhanced security but also achieves time-to-accuracy performance comparable to insecure selection methods, while incurring minimal network cost.

\section*{Acknowledgments}

We thank anonymous reviewers for their insightful comments.
We are grateful to Shaohuai Shi for providing GPU clusters to support the simulation experiments.
This research was supported in part by an RGC RIF grant under the contract R6021-20, RGC CRF grants under the contracts C7004-22G and C1029-22G, and RGC GRF grants under the contracts 16209120, 16200221, 16207922, and 16211123.

\bibliographystyle{plain}
\bibliography{main}

\appendix
\section{Proofs for Security Analysis Results (\cref{sec:design_security})}
~\label{sec:appendix_security}

\begin{reptheorem}{thm:security_random}
	If the protocol proceeds without abortion, for any $\eta > 1$, the probability that the proportion of dishonest participants, i.e., $x/s$, exceeds that in the population, i.e., $c/n$, is upper bounded as
\begin{equation*}
	\Pr [\frac{x}{s} > \eta \frac{c}{n}] \leq 1 - \sum_{i=0}^{\lfloor \eta  c s / n \rfloor} \binom{c}{i}\left(\frac{1}{m} \lfloor \frac{\alpha sm}{n_{min}} \rfloor\right)^i \left(1 - \frac{1}{m} \lfloor \frac{\alpha sm}{n_{min}} \rfloor\right)^{c - i}.
\end{equation*}
\end{reptheorem}

\begin{proof}
	The probability that $x$ exceeds some value $l$ is $\textrm{Pr}[x > l] = 1 - \sum_{i=0}^{\lfloor l \rfloor} \binom{c}{i} p^{i} (1-p)^{c-i}$, where $p$ is a dishonest client's chance of being selected into $P$.
	With client verification and consistency check enforced in \system{},  we have that $p = \frac{1}{m} \lfloor \frac{\alpha sm}{n} \rfloor \leq \frac{1}{m} \lfloor \frac{\alpha sm}{n_{min}} \rfloor$.
	The theorem follows when we combine the two arguments with $l = \lfloor \eta c s / n \rfloor$.
\end{proof}

\begin{repcorollary}{cor:security_secagg}
 	Let $t$ be the aggregation threshold of SecAgg.
	With \system{}'s random selection, the probability of the server being able to observe an individual update of an honest client, which indicates a failure of SecAgg, is upper bounded as
	\begin{equation*}
		\Pr [Fail] \leq 1 - \sum_{i=0}^{2t- s -1} \binom{c}{i}\left(\frac{1}{m} \lfloor \frac{\alpha sm}{n_{min}} \rfloor\right)^i \left(1 - \frac{1}{m} \lfloor \frac{\alpha sm}{n_{min}} \rfloor\right)^{c - i}.
	\end{equation*}
\end{repcorollary}

\begin{proof}
	In SecAgg, the secrecy of an individual update is held only when the number of dishonest participants $x < 2t - s$ (see Theorem 6.5 in~\cite{bonawitz2017practical}).
	The corollary follows when we replace $l$ with $2t - s -1$ in the proof of Theorem~\ref{thm:security_random}.
\end{proof}

\begin{repcorollary}{cor:security_ddp}
	Let $\pi$ be the distributed DP protocol parameterized by $\sigma$ and built atop SecAgg with aggregation threshold $t$.
	Given target $\delta$, an $R$-round FL training with $\pi$ and \system{}'s  random selection (\texttt{Client-Centric}) achieves $(\epsilon, \delta)$-DP, where
	\begin{equation*}
		\begin{aligned}
			\epsilon &= \min_{\substack{0 \leq k \leq \min \{c, 2t - s - 1\} \\ 0 \leq r \leq R \\ p_k + q_r < \delta}} E_\pi(s, k, r, \sigma, 1 - \frac{1-\delta}{1 - p_k - q_r}),\\
			p_k &= 1 - \left( \sum_{i=0}^k \binom{c}{i} \left(\frac{1}{m} \lfloor \frac{\alpha s m}{n_{min}} \rfloor \right)^i \left(1 - \frac{1}{m}\lfloor \frac{\alpha s m}{n_{min}} \rfloor \right)^{c - i} \right)^R,
		\end{aligned}
	\end{equation*}
	with $\sigma$ the noise multiplier used in $\pi$, and $E_\pi(\cdot)$ the privacy accounting method of $\pi$.
	In addition, $q_r = \phi_{R, r}$ with the recurrence definition of $\phi$ being $\phi_{j, r} = 1 - \gamma(j, r) + \gamma(j, r) \phi_{j - 1, r}$ with boundary condition $\phi_{r, r} = 0$ and $\gamma(j, r) = $
	\begin{equation*}
		\left( \sum_{i=0}^{r - 1} \binom{j - 1}{i} \left(\frac{1}{m} \lfloor \frac{\alpha s m}{n_{min}} \rfloor \right)^i \left(1 - \frac{1}{m}\lfloor \frac{\alpha s m}{n_{min}} \rfloor \right)^{j - 1 - i} \right)^s.
	\end{equation*}
\end{repcorollary}

\begin{proof}
	In each round, each client flips a coin and becomes a candidate with probability $p = \frac{1}{m}\lfloor\frac{\alpha s m}{n_{min}}\rfloor$.
	The server then selects $s$ participants from the candidates according to some selection policy. 
	First, consider the event that in each round, there are at most $k$ colluding participants in the candidate sets (Event $A$). The probability of $A$ is
	\begin{equation*}
		\Pr[A] = \left(\sum_{i = 0}^k \binom{c}{i}p^i (1 - p)^{c - i}\right)^R = 1 - p_k,
	\end{equation*}
	when the FL algorithm runs for $R$ rounds, and $p = \frac{1}{m}\lfloor \frac{\alpha s m}{n_{min}} \rfloor$ given the use of \system{} (\texttt{Client-Centric}). 
	
	Next, consider the event that every client is selected in at most $r$ rounds (Event $B$). We now derive a lower bound on its probability. Let $f_{n, R, r}$ be the maximum (among all possible policies) probability that at least one client is selected in at least $r + 1$ rounds, given that the population has $n$ clients and the training runs for $R$ rounds. Consider the $s$ clients selected by the server in the first round. Let $E$ be the event that none of these $s$ clients becomes a candidate in at least $r$ rounds in the remaining $R - 1$ rounds. Conditioning on $E$, all of them cannot be selected in at least $r + 1$ rounds. Thus, if one client is selected in at least $r + 1$ rounds, it must be either (i) one of the $n - s$ clients that are not selected in the first round, or (ii) selected in at least $r + 1$ rounds among the rest $R - 1$ rounds.
	We then have $f_{n, R, r}\le \Pr[E]f_{n - s, R - 1, r} + 1 - \Pr[E]$, where
	\begin{equation*}
		\Pr[E] = \left(\sum_{i = 0}^{r - 1} \binom{R - 1}{i}p^i(1 - p)^{R - 1 - i}\right)^s.
	\end{equation*}
	
	By definition, $f_{n, r, r} = 0$, as no clients can be selected in $r + 1$ rounds given only $r$ rounds in total.
	When $n > sR$, we can see by induction that $f_{n, R, r} \le \phi_{R, r}$, where $\phi_{R, r}$ is as defined in the corollary.
	For $n\le sR$, since $f_{n, R, r}$ is non-decreasing in $n$ (the server can simply ignore clients), we also have $f_{n, R, r}\le f_{sR + 1, R, r} \le \phi_{R, r}$.
	Thus, $\Pr[B]\ge 1 - f_{n, R, r} \ge 1 - \phi_{R, r} = 1 - q_r$.

	As SecAgg tolerates $2t - s - 1$ colluding clients, conditioning on the event $A\cap B$, the FL protocol is $(\varepsilon', \delta')$-DP for $0\le k \le \min\{c, 2t - s - 1\}, $ and $0\le r \le R$, where $\varepsilon' = E_{\pi}(s, k, r, \sigma, \delta')$. Thus the overall FL training is $(\varepsilon', 1 - (1 - \delta') \Pr[A\cap B])$-DP. The union bound implies that this happens with probability $\Pr[A\cap B] \ge 1 - p_k - q_r$. Thus for a target $\delta$, we can set $\delta' = 1 - (1 - \delta) / (1 - p_k - q_r)$. Minimizing $\varepsilon'$ over all valid $k$ and $r$ gives the desired results.
\end{proof}

\begin{table*}[h]
	\centering
	\caption{Per-round training time and network transfer cost for the server and average participating client with informed selection.}
	\label{tab:info_perf}
	\resizebox{\textwidth}{!}{%
		\begin{tabular}{@{}cccccccccccccc@{}}
			\toprule
			\multicolumn{2}{c|}{FL Application} &
			\multicolumn{4}{c|}{FEMNIST@CNN} &
			\multicolumn{4}{c|}{OpenImage@MobileNet} &
			\multicolumn{4}{c}{Reddit@Albert} \\ \midrule
			\multicolumn{1}{c|}{\multirow{2.5}{*}{Population}} &
			\multicolumn{1}{c|}{\multirow{2.5}{*}{Protocol}} &
			\multicolumn{2}{c|}{Time} &
			\multicolumn{2}{c|}{Network} &
			\multicolumn{2}{c|}{Time} &
			\multicolumn{2}{c|}{Network} &
			\multicolumn{2}{c|}{Time} &
			\multicolumn{2}{c}{Network} \\ \cmidrule(l){3-14} 
			\multicolumn{1}{c|}{} &
			\multicolumn{1}{c|}{} &
			Server &
			\multicolumn{1}{c|}{Client} &
			Server &
			\multicolumn{1}{c|}{Client} &
			Server &
			\multicolumn{1}{c|}{Client} &
			Server &
			\multicolumn{1}{c|}{Client} &
			Server &
			\multicolumn{1}{c|}{Client} &
			Server &
			Client \\ \midrule
			\multicolumn{1}{c|}{\multirow{3}{*}{100}} &
			\multicolumn{1}{c|}{Oort} &
			1.46min &
			\multicolumn{1}{c|}{0.83min} &
			64.88MB &
			\multicolumn{1}{c|}{3.9MB} &
			2.7min &
			\multicolumn{1}{c|}{2.07min} &
			64.35MB &
			\multicolumn{1}{c|}{3.87MB} &
			12.72min &
			\multicolumn{1}{c|}{6.45min} &
			958.55MB & 57.46MB
			\\
			\multicolumn{1}{c|}{} &
			\multicolumn{1}{c|}{Cli-Ctr} &
			1.53min &
			\multicolumn{1}{c|}{1.05min} &
			64.97MB &
			\multicolumn{1}{c|}{3.9MB} &
			2.84min &
			\multicolumn{1}{c|}{2.35min} &
			64.43MB &
			\multicolumn{1}{c|}{3.87MB} &
			12.72min &
			\multicolumn{1}{c|}{8.58min} &
			958.63MB & 57.46MB
			\\
			\multicolumn{1}{c|}{} &
			\multicolumn{1}{c|}{Srv-Ctr} &
			1.43min &
			\multicolumn{1}{c|}{0.83min} &
			64.89MB &
			\multicolumn{1}{c|}{3.9MB} &
			2.76min &
			\multicolumn{1}{c|}{2.14min} &
			64.36MB &
			\multicolumn{1}{c|}{3.87MB} &
			12.65min &
			\multicolumn{1}{c|}{6.51min} &
			958.56MB & 57.46MB
			\\ \midrule
			\multicolumn{1}{c|}{\multirow{3}{*}{400}} &
			\multicolumn{1}{c|}{Oort} &
			2.4min &
			\multicolumn{1}{c|}{1.33min} &
			0.26GB &
			\multicolumn{1}{c|}{3.56MB} &
			4.44min &
			\multicolumn{1}{c|}{3.34min} &
			0.25GB &
			\multicolumn{1}{c|}{3.53MB} &
			27.94min &
			\multicolumn{1}{c|}{16.46min} &
			3.75GB & 51.53MB
			\\
			\multicolumn{1}{c|}{} &
			\multicolumn{1}{c|}{Cli-Ctr} &
			2.64min &
			\multicolumn{1}{c|}{1.87min} &
			0.26GB &
			\multicolumn{1}{c|}{3.56MB} &
			4.69min &
			\multicolumn{1}{c|}{3.89min} &
			0.25GB &
			\multicolumn{1}{c|}{3.54MB} &
			28.15min &
			\multicolumn{1}{c|}{22.67min} &
			3.75GB & 51.53MB
			\\
			\multicolumn{1}{c|}{} &
			\multicolumn{1}{c|}{Srv-Ctr} &
			2.37min &
			\multicolumn{1}{c|}{1.34min} &
			0.26GB &
			\multicolumn{1}{c|}{3.56MB} &
			4.43min &
			\multicolumn{1}{c|}{3.38min} &
			0.25GB &
			\multicolumn{1}{c|}{3.53MB} &
			27.79min &
			\multicolumn{1}{c|}{16.51min} &
			3.75GB & 51.53MB
			\\ \midrule
			\multicolumn{1}{c|}{\multirow{3}{*}{700}} &
			\multicolumn{1}{c|}{Oort} &
			3.69min &
			\multicolumn{1}{c|}{2.09min} &
			0.45GB &
			\multicolumn{1}{c|}{3.69MB} &
			5.71min &
			\multicolumn{1}{c|}{4.12min} &
			0.45GB &
			\multicolumn{1}{c|}{3.66MB} &
			40.83min &
			\multicolumn{1}{c|}{25.22min} &
			6.56GB & 52.57MB
			\\
			\multicolumn{1}{c|}{} &
			\multicolumn{1}{c|}{Cli-Ctr} &
			3.83min &
			\multicolumn{1}{c|}{2.79min} &
			0.46GB &
			\multicolumn{1}{c|}{3.7MB} &
			6.01min &
			\multicolumn{1}{c|}{4.98min} &
			0.45GB &
			\multicolumn{1}{c|}{3.67MB} &
			40.97min &
			\multicolumn{1}{c|}{34.83min} &
			6.56GB & 52.57MB
			\\
			\multicolumn{1}{c|}{} &
			\multicolumn{1}{c|}{Srv-Ctr} &
			3.35min &
			\multicolumn{1}{c|}{1.95min} &
			0.45GB &
			\multicolumn{1}{c|}{3.69MB} &
			5.37min &
			\multicolumn{1}{c|}{3.98min} &
			0.45GB &
			\multicolumn{1}{c|}{3.66MB} &
			39.8min &
			\multicolumn{1}{c|}{24.73min} &
			6.56GB & 52.57MB 
			\\ \bottomrule
		\end{tabular}%
	}
\end{table*}

\section{Client-Centric v.s. Server-Centric}
~\label{sec:appendix_centric}

We consider \texttt{Client-Centric} as the primary approach for \system{}'s random selection as it offers superior security  (\cref{sec:design_random}).
We illustrate this with multiple-round DP training, where the privacy cost is contingent upon the client selected for the maximum number of times, assuming a fixed noise multiplier.
With \texttt{Client-Centric}, the server cannot predict which client will be self-sampled the most frequently until the training stops.
Therefore, the server cannot reserve this specific client for each round it reports to join to construct the worst-case privacy leakage.
However, with \texttt{Server-Centric}, the server can predict the selection outcomes and reserve this client with certainty.
Hence, \texttt{Server-Centric} provides the adversary with less randomness,
as implied by the following corollary:

\begin{corollary}[Controlled Privacy Cost in Distributed DP, II]
	Let $\pi$ be the distributed DP protocol parameterized by $\sigma$ and built atop SecAgg with aggregation threshold $t$.
	Given target $\delta$, an $R$-round FL training with $\pi$ and \system{}'s  random selection (\texttt{Server-Centric}) achieves $(\epsilon, \delta)$-DP, where
	\begin{equation*}
		\begin{aligned}
			\epsilon &= \min_{\substack{0 \leq k \leq \min \{c, 2t - s - 1\} \\ 0 \leq r \leq R, p_k + q_r < \delta}} E_\pi(s, k, r, \sigma, 1 - \frac{1-\delta}{1 - p_k - q_r}),\\
			p_k &= 1 - \left( \sum_{i=0}^k \binom{c}{i} \left(\frac{1}{m} \lfloor \frac{\alpha s m}{n_{min}} \rfloor \right)^i \left(1 - \frac{1}{m}\lfloor \frac{\alpha s m}{n_{min}} \rfloor \right)^{c - i} \right)^R,\\
			q_r &= 1 - \left( \sum_{i=0}^r \binom{R}{i} \left(\frac{1}{m} \lfloor \frac{\alpha s m}{n_{min}} \rfloor \right)^i \left(1 - \frac{1}{m}\lfloor \frac{\alpha s m}{n_{min}} \rfloor \right)^{R - i} \right)^{n_{max}},
		\end{aligned}
	\end{equation*}
	with $\sigma$ the noise multiplier and $E_\pi(\cdot)$ the privacy accounting method for $\pi$, and $n_{max}$ the maximum possible size of the population.
	\label{cor:security_ddp_server}
\end{corollary}

\begin{proof}
	The proof closely resembles that of Corollary~\ref{cor:security_ddp}, with the exception of the calculation of $1 - q_r$, which represents a lower bound on the probability that all clients in the population are selected for at most $r$ rounds (Event $B$).
	Considering the server's predictability, we could not get a bound as tight as that of \texttt{Client-Centric}.
	Instead, the bound is
	\begin{equation*}
		\begin{aligned}
			\Pr[B] &= \left( \sum_{i=0}^r \binom{R}{i} \left(\frac{1}{m} \lfloor \frac{\alpha s m}{n_{min}} \rfloor \right)^i \left(1 - \frac{1}{m}\lfloor \frac{\alpha s m}{n_{min}} \rfloor \right)^{R - i} \right)^n \\
			&\geq \left( \sum_{i=0}^r \binom{R}{i} \left(\frac{1}{m} \lfloor \frac{\alpha s m}{n_{min}} \rfloor \right)^i \left(1 - \frac{1}{m}\lfloor \frac{\alpha s m}{n_{min}} \rfloor \right)^{R - i} \right)^{n_{max}} \\
			&= 1 - q_r.
		\end{aligned}
	\end{equation*}
	The corollary follows when one substitutes this lower bound of $\Pr[B]$ with that used in the proof of Corollary~\ref{cor:security_ddp} and follows the remainder of that proof.
\end{proof}

\begin{figure}[t]
	\centering
	\begin{subfigure}[b]{1.0\columnwidth}
		\centering
		\includegraphics[width=\columnwidth]{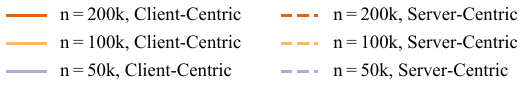}
	\end{subfigure}
	\begin{subfigure}[b]{1.0\columnwidth}
		\centering
		\includegraphics[width=\columnwidth]{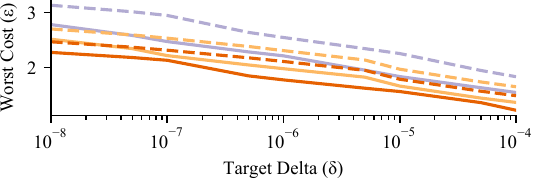}
	\end{subfigure}
	\caption{\texttt{Client-Centric} provides stronger DP privacy guarantees (lower $\epsilon$) compared to \texttt{Server-Centric}.}
	\label{fig:cor_2_cmp}
\end{figure}

To visualize the difference between the bounds presented in Corollary~\ref{cor:security_ddp} and~\ref{cor:security_ddp_server}, we refer to the example discussed in~\cref{sec:design_security}.
We set $n_{max}$ to $10^9$ for \texttt{Server-Centric}, which we consider to be the physical limit for the number of devices in the world.
Figure~\ref{fig:cor_2_cmp} illustrates the gap in their privacy cost, demonstrating a slight disadvantage of \texttt{Server-Centric} across all scenarios.
For instance, in the case of dealing with a population of size $2 \times 10^5$ when employing the same noise multiplier during training on FEMNIST, \texttt{Client-Centric} achieves $\epsilon=1.8$, whereas \texttt{Server-Centric} results in a higher cost of $\epsilon=2.3$.

\section{Execution Efficiency of Informed Selection}
~\label{sec:appendix_informed}

Table~\ref{tab:info_perf} presents the per-round time and network cost measured at both the server and average participant with informed selection.
Compared to Oort, \system{} implemented with \texttt{Client-Centric} incurs a time cost of less than 5\%, 10\%, and 5\% for population sizes of 100, 400, and 700, respectively.
The overhead for \texttt{Server-Centric} is 2\%, 0\%, and 0\%, respectively.
Regarding the network footprint, \system{} consistently introduces less than 1\% additional cost.
In all, \system{} achieves an acceptable time cost and negligible network overhead in both random selection (\cref{sec:eval_efficiency}) and informed selection.

\end{document}